\documentclass[11pt]{article}
\usepackage{latexsym}
\usepackage{a4wide}
\usepackage{amssymb,amsmath}
\usepackage{amsthm}
\usepackage{amsfonts}

\newtheorem{proposition}{Proposition}
\newtheorem{theorem}{Theorem}
\newtheorem{corollary}{Corollary}
\newtheorem{claim}{Claim}[theorem]
\newtheorem{lemma}[claim]{Lemma}

\theoremstyle{definition}

\newcommand{\auf}{\langle}
\newcommand{\zu}{\rangle}

\newcommand{\Wh}{W_h}
\newcommand{\Wv}{W_v}
\newcommand{\Rh}{R_h}
\newcommand{\Rv}{R_v}
\newcommand{\Rvr}{R_v^+}

\newcommand{\mprod}{\times}
\newcommand{\dprod}{\times^\delta}
\newcommand{\sqprod}{\times^\delta_{\textit{sq}}}     
\newcommand{\sqfprod}{\times^\delta_{\textit{sqf}}}

\newcommand{\OnSetu}[1]{\Sigma_{#1}^\irun}   
\newcommand{\OnSetd}[1]{\Sigma_{#1}^\lrun}   
\newcommand{\linOnSetd}{\Delta^\lrun_i}
\newcommand{\linOnSetu}{\Delta^\irun_i}

\newcommand{\sub}{\textit{sub}}
\newcommand{\hd}{\textit{hd}}
\newcommand{\vd}{\textit{vd}}

\newcommand{\D}{\Diamond}
\newcommand{\B}{\Box}
\newcommand{\Bh}{\B_{h}}
\newcommand{\Dh}{\D_{h}}
\newcommand{\Bv}{\B_{v}}
\newcommand{\Dv}{\D_{v}}

\newcommand{\Uv}{\Bv^+}
\newcommand{\Uh}{\Bh^+}

\newcommand{\Ev}{\Dv^+}

\newcommand{\Bhr}{\blacksquare_{h}}

\newcommand{\finci}{c_i^{++}}
\newcommand{\fdeci}{c_ i^{--}}
\newcommand{\ftest}{c_i^{??}}
\newcommand{\fincj}{c_j^{++}}
\newcommand{\fdecj}{c_ j^{--}}
\newcommand{\ftestj}{c_j^{??}}
\newcommand{\lrun}{\circ}
\newcommand{\irun}{\bullet}
\newcommand{\stepi}{\mathop{\to}^{\alpha}}
\newcommand{\stepin}{\mathop{\to}^{\alpha_n}}
\newcommand{\lstepin}{\mathop{\to}_{\textit{\scriptsize lossy}}^{\alpha_n}}
\newcommand{\istepin}{\mathop{\to}_{\textit{\scriptsize i\_err}}^{\alpha_n}}

\newcommand{\lstepi}{\mathop{\to}_{\textit{\scriptsize lossy}}^{\alpha}}
\newcommand{\istepi}{\mathop{\to}_{\textit{\scriptsize i\_err}}^{\alpha}}

\newcommand{\init}{\mathsf{ini}}
\newcommand{\fin}{\mathsf{fin}}
\newcommand{\qini}{q_{\init}}
\newcommand{\ierror}{insertion-error}

\newcommand{\id}{\mathsf{id}}
\newcommand{\univf}{\mathsf{univ}^\delta}

\newcommand{\gridfw}{\mathsf{grid}}
\newcommand{\gridfwfin}{\mathsf{grid}^{\textit{fin}}}

\newcommand{\fixu}{\mathsf{fix}^\irun}   
\newcommand{\incu}{\mathsf{inc}^\irun}   
\newcommand{\decu}{\mathsf{dec}^\irun}  
\newcommand{\fixd}{\mathsf{fix}^\lrun}  
\newcommand{\incd}{\mathsf{inc}^\lrun}   
\newcommand{\decd}{\mathsf{dec}^\lrun}   
\newcommand{\lexei}{\mathsf{do}^\lrun(\alpha)}
\newcommand{\lexein}{\mathsf{do}^\lrun(\alpha_n)}
\newcommand{\lexeinp}{\mathsf{do}^\lrun(\alpha_{n+1})}
\newcommand{\iexei}{\mathsf{do}^\irun(\alpha)}
\newcommand{\iexein}{\mathsf{do}^\irun(\alpha_n)}
\newcommand{\iexeinp}{\mathsf{do}^\irun(\alpha_{n+1})}
\newcommand{\mfw}{\varphi_M}
\newcommand{\mfwfin}{\varphi_M^{\textit{fin}}}
\renewcommand{\stop}{\mathsf{end}}
\newcommand{\St}{\widehat{\mathsf{S}}}  

\newcommand{\lingridfw}{\mathsf{lingrid}}
\newcommand{\couf}{\xi_M}
\newcommand{\linfixu}{\mathsf{lin\_fix}^\irun}
\newcommand{\linincu}{\mathsf{lin\_inc}^\irun}
\newcommand{\lindecu}{\mathsf{lin\_dec}^\irun}
\newcommand{\linfixd}{\mathsf{lin\_fix}^\lrun}
\newcommand{\linincd}{\mathsf{lin\_inc}^\lrun}
\newcommand{\lindecd}{\mathsf{lin\_dec}^\lrun}
\newcommand{\linlexei}{\mathsf{lin\_do}^\lrun(\alpha)}
\newcommand{\liniexei}{\mathsf{lin\_do}^\irun(\alpha)}
\newcommand{\linlexein}{\mathsf{lin\_do}^\lrun(\alpha_n)}
\newcommand{\linlexeinp}{\mathsf{lin\_do}^\lrun(\alpha_{n+1})}
\newcommand{\liniexein}{\mathsf{lin\_do}^\irun(\alpha_{n})}
\newcommand{\liniexeinp}{\mathsf{lin\_do}^\irun(\alpha_{n+1})}
\newcommand{\linmfw}{\psi_M}

\newcommand{\lexeibw}{\mathsf{bw\_do}^\lrun(\alpha)}
\newcommand{\iexeibw}{\mathsf{bw\_do}^\irun(\alpha)}
\newcommand{\fixubw}{\mathsf{bw\_fix}^\irun}   
\newcommand{\incubw}{\mathsf{bw\_inc}^\irun}   
\newcommand{\decubw}{\mathsf{bw\_dec}^\irun}  
\newcommand{\fixdbw}{\mathsf{bw\_fix}^\lrun}  
\newcommand{\incdbw}{\mathsf{bw\_inc}^\lrun}   
\newcommand{\decdbw}{\mathsf{bw\_dec}^\lrun}   

\newcommand{\pvar}{\mathsf{P}}
\newcommand{\rvar}{\mathsf{R}}
\newcommand{\cminus}{\mathsf{C}_i^\lrun}
\newcommand{\cplus}{\mathsf{C}_i^\irun}
\newcommand{\linminus}{\mathsf{In}_i^\lrun}
\newcommand{\loutminus}{\mathsf{Out}_i^\lrun}
\newcommand{\linplus}{\mathsf{In}_i^\irun}
\newcommand{\loutplus}{\mathsf{Out}_i^\irun}
\newcommand{\tvar}{\mathsf{tick}}

\newcommand{\F}{\mathfrak{F}}
\newcommand{\G}{\mathfrak{G}}
\newcommand{\HH}{\mathfrak{H}}
\newcommand{\M}{\mathfrak{M}}
\newcommand{\Mfw}{\M_\infty}
\newcommand{\efw}{\mu}
\newcommand{\N}{\mathfrak{N}}
\newcommand{\linMfw}{\N_\infty}
\newcommand{\linefw}{\nu}
\newcommand{\C}{\mathcal{C}}
\newcommand{\Ch}{\C_h}
\newcommand{\Cv}{\C_v}
\newcommand{\Fr}{\mathsf{Fr}\,}

\newcommand{\Log}{\mathsf{Logic\_of}}
\newcommand{\disKfourt}{\mathbf{DisK4.3}}
\newcommand{\Diff}{\mathbf{Diff}}
\newcommand{\Kfour}{\mathbf{K4}}
\newcommand{\Kfourt}{\mathbf{K4.3}}
\newcommand{\Sfive}{\mathbf{S5}}
\newcommand{\K}{\mathbf{K}}
\newcommand{\T}{\mathbf{T}}

\newcommand{\GLt}{\mathbf{GL.3}}

\newcommand{\Grzt}{\mathbf{Grz.3}}

\newcommand{\Sfourt}{\mathbf{S4.3}}
\newcommand{\Alt}{\mathbf{Alt}}
\newcommand{\DAlt}{\mathbf{DAlt}}

\begin{document}

\title{The decision problem of modal product logics\\ with a diagonal, and faulty counter machines}

\author{C. Hampson$^1$, S.Kikot$^2$, and A. Kurucz$^1$\\[10pt]
{\small ${}^1$Department of Informatics}\\
{\small King's College London, U.K.}\\[7pt]
{\small ${}^2$Institute for Information Transmission Problems}\\
{\small Moscow Institute for Physics and Technology}\\
{\small Moscow, Russia}}


\maketitle


\begin{abstract}
In the propositional modal (and algebraic) treatment of two-variable first-order logic equality
is modelled by a `diagonal' constant, interpreted in square products of universal frames as the 
identity (also known as the `diagonal') relation.
Here we study the decision problem of  products of two \emph{arbitrary} modal logics equipped with
such a diagonal. As the presence or absence of equality in two-variable first-order logic does not 
influence the complexity of its satisfiability problem, one might expect that
adding a diagonal to product logics in general is similarly harmless.
We show that this is far from being the case, and there can be quite a big jump in complexity,
even from decidable to the highly undecidable.
Our undecidable logics can also be viewed as new fragments of first-order logic
where adding equality changes a decidable fragment to undecidable.
We prove our results by a novel application of counter machine problems.
While our formalism apparently cannot force reliable counter machine computations directly, the presence of
a unique diagonal in the models makes it possible to encode both lossy and
\ierror\ computations, for the \emph{same} sequence of instructions. We show that,
given such a pair of faulty computations, it is then possible to reconstruct a reliable run from them.
\end{abstract}



\section{Introduction}\label{intro}

It is well-known that the first-order quantifier $\forall x$ can be considered as an `$\Sfive$-box': a
propositional modal $\Box$-operator interpreted over universal frames (that is, relational structures 
$\auf W,R\zu$ where $R=W\mprod W$).
The so-called `standard translation', mapping modal formulas to first-order ones,
establishes a validity preserving, bijective connection between the modal logic
$\Sfive$ and the one-variable fragment of classical first-order logic \cite{Wajsberg33}.
The idea of generalising such a propositional approach to full first-order logic was suggested and
thoroughly investigated both in modal setting \cite{Quine71,Kuhn80,Venema91},
and in algebraic logic \cite{Halmos62,Henkinetal85}.
In particular,
the bimodal logic $\Sfive\mprod\Sfive$ over two-dimensional (2D) \emph{squares}
of universal frames corresponds to the equality and substitution free fragment of two-variable first-order logic, via a translation that maps propositional variables $\pvar$ to binary predicates $\pvar(x,y)$, the modal boxes $\Box_0$ and $\Box_1$ to the first-order quantifiers $\forall x$ and $\forall y$, and the Boolean connectives to themselves.
In this setting, \emph{equality} between the two first-order variables can be modally `represented' by 
extending the bimodal language with a constant $\delta$,  
interpreted in square frames with universe $W\times W$ as the \emph{diagonal} set
\[
\{\auf x,x\zu : x\in W\}.
\]
The resulting three-modal logic (algebraically, representable 2D 
cylindric algebras \cite{Henkinetal85})
is now closer to the full two-variable fragment
(though $\pvar(y,x)$-like transposition of variables is still not expressible in it).
The generalisation of the modal treatment of full two-variable first-order logic to \emph{products}
of two arbitrary modal logics equipped with a diagonal constant (together
with modal operators `simulating' the substitution and transposition of first-order variables) was suggested in \cite{Segerberg73,Shehtman78}.
The product construction as a general combination method on modal logics was introduced in 
\cite{Gabbay&Shehtman98}, and has been extensively studied ever since (see \cite{gkwz03,Kurucz07} for surveys and references).
Two-dimensional product logics can
 not only be regarded as generalisations of the first-order quantifiers \cite{Kurucz13}, but
 they are also connected to several other 
logical formalisms, such as the one-variable fragment of modal and temporal logics, modal and temporal description logics, and spatio-temporal logics.
At first sight, the diagonal constant can only be meaningfully used
in applications where the domains of the two component frames consist of
objects of similar kinds, or at least overlap. However, as modal
languages cannot distinguish between isomorphic frames, in fact \emph{any}
subset $D$ of a Cartesian product $\Wh \times \Wv$ can be considered as an interpretation
of the diagonal constant, as long as it is both `\emph{horizontally}' and `\emph{vertically}' \emph{unique} in the following sense:
\begin{align}
\label{unique1}
&\forall x\in \Wh,\,\forall y,y'\in \Wv\
\bigl(\auf x,y\zu,\auf x,y'\zu\in D\ \to\ y=y'\bigr),\\
\label{unique2}
&\forall x,x'\in \Wh,\,\forall y\in Wv\
\bigl(\auf x,y\zu,\auf x',y\zu\in D\ \to\  x=x'\bigr).
\end{align}
So, say, 
in the one-variable constant-domain fragment of first-order temporal (or modal) logics, the diagonal constant can be added in order
to single out a set of special `time-stamped' objects of the domain, provided no special object is chosen twice and
at every moment of time (or world along the modal accessibility relation) at most one special object is chosen.

In this paper we study the decision problem of $\delta$-\emph{product logics\/}: arbitrary 2D product logics equipped with a diagonal.
It is well-known that
the presence or absence of equality in the two-variable fragment of first-order logic does not 
influence the {\sc coNExpTime}-completeness of its validity problem
\cite{Scott62,Mortimer75,Graedeletal97}.
 So one might expect that
adding a diagonal to product logics in general is similarly harmless.
The more so that decidable product logics like $\K\mprod\K$ (the bimodal logic of all product frames) remain decidable when one adds
modal operators `simulating' the substitution and transposition of first-order variables \cite{Shehtman11}.
However, we show that adding the diagonal is more dangerous,
and there can be quite a big jump in complexity.
In some cases, the global consequence relation of product logics can be reduced the validity-problem
of the corresponding $\delta$-products (Prop.~\ref{p:globalred}).
We also show (Theorems~\ref{t:kundec}, \ref{t:linundec}) 
that if $L$ is any logic having an infinite rooted frame where each point can be
accessed by at most one step from the root, then both $\K\dprod L$ and $\Kfourt\dprod L$ 
are undecidable (here $\K$ is the unimodal logic of all frames, and $\Kfourt$ is the unimodal
logic of linear orders).
Some notable consequences of these results are:
\begin{enumerate}
 \item[(i)]
 $\K\dprod\Sfive$ is undecidable,
 (while $\K\mprod\Sfive$ is {\sc coNExpTime}-complete \cite{Marx99}, and even the global
 consequence relation of $\K\mprod\Sfive$ is decidable in {\sc co2NExp\-Time} \cite{Wolter99,Schmidt&Tishkovsky02}).
\item[(ii)]
$\Kfourt\dprod\Sfive$ is undecidable (while $\Kfourt\mprod\Sfive$ is decidable in {\sc 2ExpTime}
\cite{Reynolds97}).
 \item[(iii)]
 $\K\dprod\K$ is undecidable
 (while $\K\mprod\K$ is decidable \cite{Gabbay&Shehtman98}, though not in {\sc ElementaryTime}
 \cite{GollerJL12}).
\end{enumerate}
See also Table~\ref{t:results} for some known results on product logics, and how our present results on
$\delta$-products compare with them.
\begin{table}
\begin{center}
\begin{tabular}{l||l|l|l|}
& & {\small global} & \\
&  {\small validity of} & \ {\small consequence of} & {\small\bf validity of} \\
& \qquad\qquad{\small product logic} & \ {\small product logic} & \,{\small\bf $\delta$-product logic}\\[5pt]
\hline\hline
&&&\\[-5pt]
& {\small {\sc coNExpTime}-complete} &  {\small same} & {\small {\sc coNExpTime}-}\\
{\small $\Sfive\mprod\Sfive$} & \ \  {\small \cite{Scott62,Mortimer75,Graedeletal97,Marx99}} & \ \ {\small as validity} & \ \ {\small complete}\\
&  && \ \ {\small \cite{Scott62,Mortimer75,Graedeletal97}}\\[5pt]
\hline
&&&\\[-5pt]
& {\small {\sc coNExpTime}-complete} &  {\small decidable in} & {\small\bf undecidable}\\
{\small $\K\mprod\Sfive$} & \ \ {\small  \cite{Marx99}} & \ {\small {\sc co2NExp\-Time}} &
 \ \ {\small\bf Cor.~\ref{co:kundec}}\\
 && \ {\small \cite{Wolter99,Schmidt&Tishkovsky02}} &\\[5pt]
\hline
&&&\\[-5pt]
& {\small decidable \cite{Gabbay&Shehtman98}} &  & {\small\bf undecidable}\\
{\small $\K\mprod\K$} & {\small not in {\sc ElementaryTime}} & {\small undecidable \cite{Marx99}} & \ \ {\small\bf Cor.~\ref{co:globalred}} \\
& \ {\small \cite{GollerJL12}} &&\\ [5pt]
\hline
&&&\\[-5pt]
& {\small decidable} & {\small same} & {\small\bf undecidable}\\
{\small $\Kfourt\mprod\Sfive$} & \ \ {\small in {\sc 2ExpTime} \cite{Reynolds97}}& \ \ {\small as validity} & \ \ {\small\bf Cor.~\ref{co:linundec}}\\
& {\small {\sc coNExpTime}-hard \cite{Marx99}} &  &\\[5pt]
\hline
&&&\\[-5pt]
 &  {\small decidable} & {\small same}  & \\
{\small $\Kfour\mprod\Sfive$} & \ \ {\small in {\sc coN2ExpTime} \cite{Gabbay&Shehtman98}} & \ \ {\small as validity} & \qquad\quad \framebox{{\bf ?}}\\
& {\small {\sc coNExpTime}-hard \cite{Marx99}} & &\\[5pt]
\hline
&&&\\[-5pt]
& {\small decidable \cite{Wolter99}} & {\small undecidable \cite{Gutierrez-BasultoJ014}} & {\small\bf undecidable}\\
{\small $\Kfour\mprod\K$} &  {\small not in {\sc ElementaryTime}} && \ \  {\small\bf Cor.~\ref{co:globalred}} \\
& \ {\small \cite{GollerJL12}} && \\[5pt]
\hline
&&&\\[-5pt]
{\small $\Kfour\mprod\Kfour$} & {\small undecidable \cite{gkwz05a}} & {\small same} & {\small\bf undecidable} \\
& & \ \ {\small as validity} & \ \  {\small\bf Prop.~\ref{p:cons}} \\[3pt]
\hline
&&&\\[-5pt]
& {\small decidable} & & {\small\bf decidable}\\
{\small $\K\!\mprod\!\Alt(n)$} &  \;{\small in {\sc coNExpTime} ($n>1$)} & {\small undecidable} & {\small {\bf in} {\sc coNExpTime}} \\
& \;{\small in {\sc ExpTime} ($n=1$) \cite{gkwz03}} & & \ \ {\small\bf Thm.~\ref{t:dec}}\\[5pt]
\hline
\end{tabular}
\caption{Product vs. $\delta$-product logics.}\label{t:results}
\end{center}
\end{table}

While all the above $\delta$-product logics are recursively enumerable (Theorem~\ref{t:re}),
we also show that in some cases decidable product logics can turn highly undecidable by adding
a diagonal. For instance, both $\K\dprod\Sfive$ and $\K\dprod\K$ when restricted to finite (but unbounded) product frames
result in non-recursively enumerable logics (Theorem~\ref{t:kundecfin}). Also,
$\Log\auf\omega,<\zu\dprod\Sfive$ is $\Pi_1^1$-hard (Theorem~\ref{t:linundecdisc}).
On the other hand, the unbounded width of the second-component frames 
seems to be essential in obtaining these results. Adding a diagonal to decidable product logics of the form $\K\mprod\Alt(n)$,
$\Sfive\mprod\Alt(n)$, and $\Alt(m)\mprod\Alt(n)$ result in decidable logics, sometimes even with the same
upper bounds that are known for
the products (Theorems~\ref{t:dec} and \ref{t:altnp}) (here $\Alt(n)$ is the unimodal logic of
frames where each point has at most $n$ successors for some $0<n<\omega$).

Our undecidable $\delta$-product logics can also be viewed as new fragments of first-order logic
where adding equality changes a decidable fragment to undecidable. (A well-known such 
fragment is the G\"odel class \cite{Godel33,Goldfarb84}.) In particular, consider the following `2D extension' of
the standard translation \cite{Gabbay&Shehtman99}, from bimodal formulas to three-variable first-order formulas having two free variables $x$ and $y$ and a
built-in binary predicate $\rvar$:
\begin{align*}
\pvar^\dag & :=\  \pvar(x,y),\quad\mbox{for propositional variables $\pvar$},\\
(\neg\phi)& :=\ \neg\phi^\dag\quad\mbox{ and }\quad (\phi\land\psi)^\dag:=\ \phi^\dag\land\psi^\dag,\\
(\Box_0\phi)^\dag & :=\  \forall z\, \bigl(\rvar(x,z)\to \phi^\dag(z/x,y)\bigr),\\
(\Box_1\phi)^\dag & :=\  \forall z\, \bigl(\rvar(y,z)\to \phi^\dag(x,z/y)\bigr).
\end{align*}
It is straightforward to see that, for any bimodal formula $\phi$, $\phi$ is satisfiable in
the (decidable)
 modal product logic $\K\mprod\K$ 
iff $\phi^\dag$ is satisfiable in first-order logic. So the image of ${}^\dag$ is a decidable fragment
of first-order logic that becomes undecidable when equality is added.

Our results show that in many cases the presence of a \emph{single} proposition (the diagonal)
with the `horizontal' and `vertical' uniqueness properties \eqref{unique1}--\eqref{unique2} is
enough to cause undecidability of 2D product logics. If each of the component logics has 
a \emph{difference operator}, then their product can express `horizontal' and `vertical' uniqueness
of \emph{any} proposition. For example, 
this is the case when each component is either the unimodal logic $\Diff$ of all frames of the 
form $\auf W,\ne\zu$, or a logic determined by strict linear orders such as $\Kfourt$ or $\Log\auf\omega,<\zu$.
So our Theorems~\ref{t:linundec} and \ref{t:linundecdisc} can be regarded as generalisations of the
undecidability results of \cite{Reynolds&Z01} on `linear'$\mprod$`linear'-type products, and those
of \cite{Hampson&Kurucz14} on `linear'$\mprod\Diff$-type products.

\paragraph{On the proof methods.}
Even if 2D product structures are always grid-like by definition, there are two issues one needs to deal with in order to encode grid-based complex problems into them:
\begin{itemize}
\item[(i)]
to generate infinity, even when some component structure is not transitive, and
\item[(ii)]
somehow to `access' or `refer to' neighbouring-grid points, even when there is no `next-time' operator in the language, and/or the component structures are transitive or even universal.
\end{itemize}
When both component structures are transitive, then (i) is not a problem.
If in addition component structures of arbitrarily large depths are available, then 
(ii) is usually solved
by `diagonally' encoding the $\omega\times\omega$-grid, and then use reductions of tiling or Turing machine problems \cite{Marx&Reynolds99,Reynolds&Z01,gkwz05a}.
When both components can express the uniqueness of any proposition (like strict linear orders or the difference operator), then it is also possible to make direct use of the grid-like nature of product structures and obtain undecidability by forcing reliable counter machine computations \cite{Hampson&Kurucz14}.
However, $\delta$-product logics of the form $L\dprod\Sfive$
apparently neither can force such computations directly, nor they can diagonally encode the $\omega\times\omega$-grid. Instead,
we prove our lower bound results by a novel application of counter machine problems.
The presence of
a unique diagonal in the models makes it possible to encode both \emph{lossy} and
\emph{\ierror\/} computations, for the \emph{same} sequence of instructions.
We then show (Prop.~\ref{p:approx}) that,
given such a pair of faulty computations, one can actually reconstruct a reliable run from them.
The upper bound results are shown by a straightforward selective filtration.

\medskip
The structure of the paper is as follows.
Section~\ref{dprod} provides all the necessary definitions. 
In Section~\ref{conn}
we establish connections between our logics and other formalisms, and discuss
some consequences of these connections on the decision problem of $\delta$-products.
In Section~\ref{cm} we introduce counter machines, and discuss how reliable counter machine computations can be approximated by faulty (lossy and
\ierror) ones. Then in Sections~\ref{kprod} and \ref{linprod}
we state and prove our undecidability results on $\delta$-products having a $\K$ or a
`linear' component, respectively. The decidability results are proved in Section~\ref{dec}.
Finally, in Section~\ref{disc} we discuss some related open problems.


\section{$\delta$-product logics}\label{dprod}

In what follows we assume that the reader is familiar with the basic notions in modal logic
and its possible world semantics (see \cite{Blackburnetal01,cz} for reference).
Below we
summarise the necessary notions and notation for our 3-modal case only, but we will use them
throughout for the uni- and bimodal cases as well.
We define our \emph{formulas} by the following grammar:
\[
\phi:=\ \pvar\mid\delta\mid\neg\phi\mid\phi\land\psi\mid\Bh\phi\mid\Bv\phi,
\]
where $\pvar$ ranges over an infinite set of propositional variables.
We use the usual abbreviations $\lor$, $\to$, $\leftrightarrow$, 
$\bot:=\pvar\land\neg\pvar$, 
$\D_i:=\neg\B_i\neg$, and also
%
\[
\D_i^+\phi:=\ \phi\lor\D_i\phi,\hspace*{3cm} \B_i^+\phi:=\ \phi\land\B_i\phi,
%
\]
for $i=h,v$. (The subscripts are indicative of the 2D intuition: $h$ for `horizontal' and
$v$ for `vertical'.)

A $\delta$-\emph{frame} is a tuple $\F=\auf W,\Rh,\Rv,D\zu$ where $R_i$ are binary relations on the non-empty set $W$, and $D$ is a subset of $W$. We call $\F$ \emph{rooted} if there is some $w$
such that $wR^\ast v$ for all $v\in W$, for the reflexive and transitive closure $R^\ast$ of $R:=\Rh\cup\Rv$.
A \emph{model based on} $\F$ is a pair $\M=\auf\F,\nu\zu$, where $\nu$ is a
function mapping propositional variables to subsets of $W$. The \emph{truth relation}
$\M,w\models\phi$ is defined, for all $w\in W$, by induction on $\phi$ as usual. In particular,
\[
\M,w\models\delta\quad\mbox{ iff }\quad w\in D.
\]
We say that $\phi$ is \emph{satisfied in} $\M$, if there is $w\in W$ with $\M,w\models\phi$.
We write $\M\models\phi$, if $\M,w\models\phi$ for every $w\in W$. 
Given a set $L$ of formulas, we write $\M\models L$ if
$\M\models\phi$ for every $\phi$ in $L$.
Given formulas $\phi$ and $\psi$,
we write $\phi\models^\ast_L\psi$ iff $\M\models\psi$ for every model 
$\M$ such that $\M\models L\cup\{\phi\}$.

We say that $\phi$ is \emph{valid in\/} $\F$, 
if $\M\models\phi$ for every model $\M$ based on $\F$.
If every formula
in a set $L$ is valid in $\F$, then we say that $\F$ is a \emph{frame for}  $L$.
We let $\Fr L$ denote the class of all frames for $L$.
For any class $\C$ of $\delta$-frames, we let
\[
\Log\,\C:=\{\phi :\phi\mbox{ is a formula valid in every member of }\C\}.
\]
We call a set $L$ of formulas a \emph{Kripke complete logic} if $L=\Log\,\C$ for some
class $\C$. A Kripke complete logic $L$ such that for all formulas $\phi$ and $\psi$,
$\phi\models^\ast_L\psi$ iff $\M\models\phi$ implies $\M\models\psi$ for every model $\M$
based on a frame for $L$, is called \emph{globally Kripke complete}.

We are interested in some special `two-dimensional' $\delta$-frames.
Given unimodal Kripke frames $\F_h=\auf \Wh,\Rh\zu$ and $\F_v=\auf \Wv,\Rv\zu$,
their \emph{product} is the bimodal frame
\[
\F_h\mprod\F_v:=\auf\Wh\times\Wv,\overline{R}_h,\overline{R}_v\zu,
\]
where $\Wh\times\Wv$ is the Cartesian product of sets $\Wh$ and $\Wv$
and the binary relations $\overline{R}_h$ and $\overline{R}_v$ are defined by taking, 
for all $x,x'\in \Wh$, $y,y'\in \Wv$,
\begin{gather*}
\auf x,y\zu \overline{R}_h \auf x',y'\zu\quad \text{ iff }\quad x\Rh x'
\mbox{ and }y=y',\\
\auf x,y\zu \overline{R}_v \auf x',y'\zu\quad \text{ iff }\quad y\Rv y'
\mbox{ and }x=x'.
\end{gather*}
The $\delta$-\emph{product} of $\F_h$ and $\F_v$ is
the $\delta$-frame
\[
\F_h\dprod\F_v:=\auf\Wh\times\Wv,\overline{R}_h,\overline{R}_v,\id\zu,
\]
where $\auf\Wh\times\Wv,\overline{R}_h,\overline{R}_v\zu=\F_h\mprod\F_v$ and
\[
\id=\{\auf x,x\zu: x\in \Wh\cap\Wv\}.
\]
For classes $\C_h$ and $\C_v$ of unimodal frames, we define
\[
\C_h\dprod\C_v=\{\F_h\dprod \F_v : \F_i\in\C_i,\mbox{ for $i=h,v$}\}.
\]
Now, for $i=h,v$, let $L_i$ be a Kripke complete unimodal logic in the language with $\Diamond_i$.
The $\delta$-\emph{product} of $L_h$ and $L_v$ is defined as 
\[
L_h\dprod L_v: =\Log\,(\Fr L_h\dprod\Fr L_v).
\]
As a generalisation of the modal approximation of two-variable first-order logic, it might
be more `faithful' to consider
\begin{multline*}
L_h\sqprod L_v: =\{\phi : \phi \mbox{ is valid in $\F_h\dprod\F_v$, for some rooted $\mathfrak F_i=\auf W_i,R_i\zu$}\\
\mbox{in $\Fr L_i,\ i=h,v$, such that  $\Wh=\Wv$}\},
\end{multline*}
or, in case $L_h=L_v=L$, even
\[
L\sqfprod L: =\{\phi : \phi \mbox{ is valid in $\F\dprod\F$, for some rooted $\F\in\Fr L$}\}.
\]
Then $\Sfive\sqprod\Sfive=\Sfive\sqfprod\Sfive$ indeed corresponds to the transposition-free fragment of
two-variable first-order logic. However, $\Sfive\dprod\Sfive$ is properly contained
in $\Sfive\sqprod\Sfive$: for instance $\Dh\delta$ belongs to the latter but not to the former.
In general, clearly we always have 
$
L_h\dprod L_v\subseteq L_h\sqprod L_v
$
and
$
L\sqprod L\subseteq L\sqfprod L,
$
whenever $L_h=L_v=L$. Also,
it is not hard to give examples when the three definitions result in three different logics.
Throughout, we formulate all our results for the $L_h\dprod L_v$ cases only, but each and every of
them holds for the corresponding $L_h\sqprod L_v$ as well (and also for $L\sqfprod L$ when
it is meaningful to consider the same $L$ as
both components).

Given a set $L$ of formulas, we are interested in the following decision problems:

\medskip
\noindent
\underline{$L$-{\sc validity:}}\ \
Given a formula $\phi$, does it belong to $L$?

\medskip
\noindent
If this problem is (un)decidable, we simply say that `$L$ is (un)decidable'.
$L$-validity is the `dual' of 

\medskip
\noindent
\underline{$L$-{\sc satisfiability:}}\ \ 
\parbox[t]{12.2cm}{Given a formula $\phi$, is there a model $\M$ such that $\M\models L$ and $\phi$ is satisfied in $\M$?}

\medskip
\noindent
Clearly, if $L=\Log\,\C$ then $L$-satisfiability is the same as

\medskip
\noindent
\underline{$\C$-{\sc satisfiability:}}\ \ 
\parbox[t]{12.2cm}{Given a formula $\phi$, is there a frame $\F\in\C$ such that $\phi$ is satisfied in a model based on $\F$?}

\medskip
\noindent
We also consider

\medskip
\noindent
\underline{{\sc Global} $L$-{\sc consequence:}}\ \ 
Given formulas $\phi$ and $\psi$, does $\phi\models^\ast_L\psi$ hold?

\paragraph{Notation.}
Our notation is mostly standard. In particular, we denote by $R^+$ the \emph{reflexive closure} of a binary relation $R$. The cardinality of a set $X$ is denoted by $|X|$.
For each natural number $k<\omega$, we also consider $k$ as the finite ordinal
$k=\{0,\dots,k-1\}$.


\section{Decidability of $\delta$-products: what to expect?}\label{conn}

To begin with, the following proposition 
is straightforward from the definitions:

\begin{proposition}\label{p:cons}
$L_h\dprod L_v$ is always a conservative
extension of $L_h\mprod L_v$. 
\end{proposition}

So it follows from the undecidability results of \cite{gkwz05a} on the
corresponding product logics that  $L_h\dprod L_v$ is undecidable, whenever both $L_h$ and $L_v$
have only \emph{transitive} frames and have frames of 
\emph{arbitrarily large depths\/}. For example, $\Kfour\dprod\Kfour$ is undecidable, where
$\Kfour$ is the unimodal logic of all transitive frames.

Next, we establish connections between  the global consequence
relation of some product logics and the corresponding $\delta$-products. To begin with,
we introduce an operation on frames that we call \emph{disjoint union with a spy-point\/}.
Given unimodal frames $\F_i=\auf W_i,R_i\zu$, $i\in I$, for some index set $I$, and a fresh point $r$, we let
\[
\bigcup^r_{i\in I}\F_i:=\auf W,R\zu,
\]
where
\begin{align*}
& W=\{r\}\cup \{\auf w,i\zu : i\in I,\ w\in W_i\},\quad\mbox{and}\quad\\
& R=\bigl\{\bigl\auf r,\auf w,i\zu\bigr\zu : w\in W_i,\ i\in I\bigr\}
\cup
\bigl\{\bigl\auf\auf w,i\zu,\auf w', i\zu\bigr\zu : w,w'\in W_i,\ wR_iw',\ i\in I\bigr\}.
\end{align*}
Note that the spy-point technique is well-known in hybrid logic \cite{Blackburn&Seligman95}.

\begin{proposition}\label{p:globalred}
If $L_h$ and $L_v$ are Kripke complete logics such that both 
$\Fr L_h$ and $\Fr L_v$ are closed under the `disjoint union with a spy-point' operation
and $L_h\mprod L_v$ is globally Kripke complete,
then the global $L_h\mprod L_v$-consequence is reducible to $L_h\dprod L_v$-validity.
\end{proposition}

\begin{proof}
We show that for all bimodal ($\delta$-free) formulas $\phi$, $\psi$,
\[
\phi\models^\ast_{L_h\mprod L_v}\psi\qquad\mbox{ iff }\qquad 
\bigl((\univf\land\Bh\Bv\phi)\to\Bh\Bv\psi \bigr)
\in L_h\dprod L_v,
\]
where 
\[
\univf:=\ \Bh\Dv\delta\land \Bh\Bh\Dv\delta\land \Bv\Dh\delta\land \Bv\Bv\Dh\delta.
\]

$\Rightarrow$:
Suppose that 
$\M,\auf r_h,r_v\zu\models \univf\land\Bh\Bv\phi\land\Dh\Dv\neg\psi$ in a model $\M$
that is based on $\F_h\dprod\F_v$, for some frames $\F_i=\auf W_i,R_i\zu$ in $\Fr L_i$, $i=h,v$.
Then there exist $x_h$, $x_v$ such that $r_h\Rh x_h$, $r_v\Rv x_v$ and 
$\M,\auf x_h,x_v\zu\models \neg\psi$.
For $i=h,v$, let $\G_i$ be the subframe of $\F_i$ generated by point $x_i$, 
and let $\N$ be the
restriction of $\M$ to $\G_h\mprod\G_v$. Then 
\begin{equation}\label{models}
\N\models L_h\mprod L_v
\quad\mbox{and}\quad
\N,\auf x_h,x_v\zu\models \neg\psi.
\end{equation}
We claim that
\begin{equation}\label{univ}
\mbox{$r_iR_i w$, for all $w$ in $\G_i$ and $i=h,v$}.
\end{equation}
Indeed, let $i=h$. We prove \eqref{univ} by induction on the smallest number $n$ of $\Rh$-steps needed
to access $w$ from $x_h$. If $n=0$ then we have $r_h\Rh x_h$. Now
suppose inductively that  \eqref{univ} holds for all $w$ in $\G_h$ that are accessible in $\leq n$ 
$\Rh$-steps
from $x_h$ for some $n<\omega$, and let $w'$ be accessible in $n+1$ $\Rh$-steps. Then there
is $w$ in $\G_h$ that is accessible in $n$ steps and $w\Rh w'$. Thus $r_h \Rh w$ by the IH, and so
$\M,\auf w',r_v\zu\models\Dv\delta$ by $\univf$. Therefore, we have $w'\in \Wv$ and $r_v\Rv w'$.
Then $\M,\auf r_h,w'\zu\models\Dh\delta$  again by $\univf$,
and so $r_h \Rh w'$ as required. The $i=v$ case is similar.

Now it follows from $\M,\auf r_h,r_v\zu\models\Bh\Bv\phi$ and \eqref{univ} that
$\N\models\phi$. Therefore, $\phi\not\models^\ast_{L_h\mprod L_v}\psi$ by \eqref{models}.

\smallskip
$\Leftarrow$:
Suppose that $\M\models\phi$ and $\M,w\models\neg\psi$ in some model $\M$ with
$\M\models L_h\mprod L_v$. As $L_h\mprod L_v$ is globally Kripke complete, we may
assume that 
$\M=\auf\F_h\mprod\F_v,\mu\zu$ for some frames $\F_i=\auf W_i,R_i\zu$ in $\Fr L_i$, $i=i,h$.
Let $\F_h^\alpha$, $\alpha<|\Wv|$, be $|\Wv|$-many copies of $\F_h$,
and $\F_v^\beta$, $\beta<|\Wh|$, be $|\Wh|$-many copies of $\F_v$.
Take some fresh point $r$ and define
\[
\G_h=\auf U_h,S_h\zu:=\bigcup^r_{\alpha<|\Wv|}\F_h^\alpha\qquad\mbox{and}\qquad
\G_v=\auf U_v,S_v\zu:=\bigcup^r_{\beta<|\Wh|}\F_v^\beta.
\]
Then by our assumption, $\G_i$ is a frame for $L_i$, for $i=h,v$.
Define a model $\N:=\auf\G_h\dprod\G_v,\nu\zu$ by taking, for all propositional variables $\pvar$,
\[
\nu(\pvar):=\bigl\{\bigl\auf\auf x,\alpha\zu,\auf y,\beta\zu\bigr\zu : \auf x,y\zu\in\mu(\pvar)\bigr\}.
\]
Then $\N,\auf r,r\zu\models \Bh\Bv\phi\land\Dh\Dv\neg\psi$.
As $|U_h|=|U_v|$ and $\Fr L_i$ is closed under isomorphic copies for $i=h,v$, we can actually assume that $U_h=U_v$, and so $\N,\auf r,r\zu\models \univf$.
\end{proof}

\begin{corollary}\label{co:globalred}
$\K\dprod\K$  and $\K\dprod\Kfour$ are both undecidable.
\end{corollary}

\begin{proof}
It is not hard to check that
the 2D product logics $\K\mprod\K$ and \mbox{$\K\mprod\Kfour$} satisfy the requirements in 
Prop.~\ref{p:globalred} (cf.\ \cite[Thm.5.12]{gkwz03} for global Kripke completeness).
A  reduction of, say, the $\omega\times\omega$-tiling problem \cite{Berger66} shows
that global $\K\mprod\K$-consequence is undecidable \cite{Marx99},
and so the undecidability of $\K\dprod\K$ follows by Prop.~\ref{p:globalred}.
It is shown in \cite{Gutierrez-BasultoJ014} that the reduction of $\Kfour$ to global $\K$-consequence \cite{Tobies01}
can be `lifted' to the product level, and so $\Kfour\mprod\Kfour$ is reducible to global 
$\K\mprod\Kfour$-consequence. Therefore, the latter is undecidable \cite{gkwz05a}, and so the 
undecidability of $\K\dprod\Kfour$ follows by Prop.~\ref{p:globalred}.
\end{proof}

Note that we can also make Prop.~\ref{p:globalred} work for logics having only \emph{reflexive}
frames by making the `spy-point' reflexive, and using a slightly different `translation':
\begin{multline*}
\phi\models^\ast_{L_h\mprod L_v}\psi\qquad\mbox{ iff }\\
\bigl((\univf\land\Bh\pvar\land\Bv\pvar\land\Bh\Bv(\neg\pvar\to\phi)\to\Bh\Bv(\neg\pvar\to\psi) \bigr)
\in L_h\dprod L_v,
\end{multline*}
where $\pvar$ is a fresh propositional variable.

However, logics having only symmetric frames (like $\Sfive$), or having only frames with bounded width (like $\Kfourt$ or $\Alt(n)$) are not closed
under the `disjoint union with a spy-point' operation, and so Prop.~\ref{p:globalred} does
not apply to their products.
It turns out that in some of these cases such a reduction is either not useful in establishing
undecidability of $\delta$-products, or does not even exist.
While global $\K\mprod\Sfive$-consequence is reducible to \mbox{${\bf PDL}\mprod\Sfive$}-validity%
\footnote{Here ${\bf PDL}$ denotes Propositional Dynamic Logic.}%
, and so decidable in {\sc co2NExpTime} \cite{Wolter99,Schmidt&Tishkovsky02},
$\K\dprod\Sfive$ is shown to be undecidable in Theorem~\ref{t:kundec} below.
While $\K\dprod\Alt(n)$ is decidable by Theorem~\ref{t:dec} below,
the undecidability of global $\K\mprod\Alt(n)$-consequence can again be shown by a 
straightforward reduction
of the $\omega\times\omega$-tiling problem.

Finally,
the following general result is a straightforward generalisation of the similar theorem of
 \cite{Gabbay&Shehtman98} on product logics. 
It is an easy consequence of the recursive enumerability of the consequence
relation of (many-sorted) first-order logic:

\begin{theorem}\label{t:re}
If $L_h$ and $L_v$ are Kripke complete logics such that both 
$\Fr L_h$ and $\Fr L_v$ 
are recursively first-order definable in the language having a binary predicate symbol, then 
$L_h\dprod L_v$ is recursively enumerable.
\end{theorem}


\section{Reliable counter machines and faulty approximations}\label{cm}

A \emph{Minsky} \cite{Minsky67} 
or  \emph{counter machine} $M$ is described by a finite set $Q$ of states, 
an initial state $\qini\in Q$,
a set $H\subseteq Q$ of terminal states, 
a finite set $C=\{c_0,\dots,c_{N-1}\}$ of counters with $N>1$,
a finite nonempty set  $I_q\subseteq \textit{Op}_C\times Q$ of instructions, for each $q\in Q-H$, where
each operation in $\textit{Op}_C$ is one of the following forms, for some $i<N$:
\begin{itemize}
\item
$\finci$ (\emph{increment counter} $c_i$ \emph{by one}),
\item
$\fdeci$ (\emph{decrement counter} $c_i$ \emph{by one}),
\item
$\ftest$ (\emph{test whether counter} $c_i$ \emph{is empty}).
\end{itemize}
For each $\alpha\in \textit{Op}_C$,
we will consider three different kinds of semantics:
\emph{reliable} (as described above), 
\emph{lossy} \cite{Mayr00} (when counters can spontaneously decrease, both before and after performing $\alpha$), and
\emph{\ierror} \cite{OuaknineW06} (when counters can spontaneously increase, both before and after performing $\alpha$).

A \emph{configuration} of $M$ is a tuple $\auf q,\vec{c}\zu$ with $q\in Q$ representing the 
current state, and an $N$-tuple
$\vec{c}=\auf c_0,\dots,c_{N-1}\zu$ of natural numbers representing the current contents of the counters. 
Given $\alpha\in \textit{Op}_C$,
we say that \emph{there is a reliable} $\alpha$-\emph{step} between 
configurations $\auf q,\vec{c}\,\zu$ and $\auf q',\vec{c}\,'\zu$ (written
$\auf q,\vec{c}\,\zu\stepi\auf q',\vec{c}\,'\zu$) iff $\auf\alpha,q'\zu\in I_q$ and
\begin{itemize}
\item
if $\alpha=\finci$ then $c_i'=c_i+1$ and $c_j'=c_j$ for $j\ne i$, $j<N$;
\item
if $\alpha=\fdeci$ then $c_i'=c_i-1$ and $c_j'=c_j$ for $j\ne i$, $j<N$; 
\item
if $\alpha=\ftest$ then $c_i'=c_i=0$ and $c_j'=c_j$ for $j<N$.
\end{itemize}
We say that \emph{there is a lossy} $\alpha$-\emph{step} between 
configurations $\auf q,\vec{c}\,\zu$ and $\auf q',\vec{c}\,'\zu$ (and we write
$\auf q,\vec{c}\,\zu\lstepi\auf q',\vec{c}\,'\zu$) iff $\auf\alpha,q'\zu\in I_q$ and
\begin{itemize}
\item
if $\alpha=\finci$ then $c_i'\leq c_i+1$ and $c_j'\leq c_j$ for $j\ne i$, $j<N$;
\item
if $\alpha=\fdeci$ then $c_i'\leq c_i-1$ and $c_j'\leq c_j$ for $j\ne i$, $j<N$; 
\item
if $\alpha=\ftest$ then $c_i'=0$ and $c_j'\leq c_j$ for $j<N$.
\end{itemize}
Finally, we say that \emph{there is an \ierror\/} $\alpha$-\emph{step} between 
configurations $\auf q,\vec{c}\,\zu$ and $\auf q',\vec{c}\,'\zu$ (written
$\auf q,\vec{c}\,\zu\istepi\auf q',\vec{c}\,'\zu$) iff $\auf\alpha,q'\zu\in I_q$ and
\begin{itemize}
\item
if $\alpha=\finci$ then $c_i'\geq c_i+1$ and $c_j'\geq c_j$ for $j\ne i$, $j<N$;
\item
if $\alpha=\fdeci$ then $c_i'\geq c_i-1$ and $c_j'\geq c_j$ for $j\ne i$, $j<N$; 
\item
if $\alpha=\ftest$ then $c_i=0$ and $c_j'\geq c_j$ for $j<N$.
\end{itemize}
Now suppose that a sequence $\vec{\tau}=\bigl\auf\auf\alpha_n,q_n\zu: 0<n<B\bigr\zu$ of instructions of $M$
is given for some $0<B\leq\omega$. We say that a sequence $\vec{\varrho}=\bigl\auf\auf q_n,\vec{c}(n)\zu : n<B\bigr\zu$ of configurations is a \emph{reliable} 
$\vec{\tau}$-\emph{run of\/} $M$ if
\begin{itemize}
\item[(i)]
$q_0=\qini$, $\vec{c}(0)=\vec{0}$, and
\item[(ii)]
$\auf q_{n-1},\vec{c}(n-1)\zu\stepin\auf q_n,\vec{c}(n)\zu$
holds for every $0<n<B$.
\end{itemize}
A \emph{reliable run} is a reliable $\vec{\tau}$-run for some $\vec{\tau}$.
Similarly, a sequence $\vec{\varrho}$ satisfying (i) is called a \emph{lossy} $\vec{\tau}$-\emph{run} 
if we have $\auf q_{n-1},\vec{c}(n-1)\zu\lstepin\auf q_n,\vec{c}(n)\zu$,
and an \emph{\ierror} $\vec{\tau}$-\emph{run} if we have
$\auf q_{n-1},\vec{c}(n-1)\zu\istepin\auf q_n,\vec{c}(n)\zu$, for every $0<n<B$.
(Note that in order to simplify the presentation, in each case we only consider runs that start at state $\qini$ with all-zero counters.)

Observe that, for any given $\vec{\tau}$, if there exists a reliable $\vec{\tau}$-run, then it is unique.
The following statement says that this unique reliable $\vec{\tau}$-run can be `approximated' by 
a $\auf$lossy,\;\ierror$\zu$-pair of $\vec{\tau}$-runs:

\begin{proposition}\label{p:approx} {\bf (faulty approximation)}\\
Given any sequence $\vec{\tau}$ of instructions, there exists a reliable $\vec{\tau}$-run iff
there exist both lossy and \ierror\ $\vec{\tau}$-runs.
\end{proposition}

\begin{proof}
The $\Rightarrow$ direction is obvious, as each reliable $\vec{\tau}$-run is both a lossy and
an \ierror\ $\vec{\tau}$-run as well. For the $\Leftarrow$ direction, suppose that 
$\vec{\tau}=\bigl\auf\auf\alpha_n,q_n\zu: 0<n<B\bigr\zu$ for some $B\leq\omega$,
$\bigl\auf\auf q_n,\vec{c}^{\;\lrun}(n)\zu : n<B\bigr\zu$ is a lossy $\vec{\tau}$-run, and
$\bigl\auf\auf q_n,\vec{c}^{\;\irun}(n)\zu : n<B\bigr\zu$ is an \ierror\ $\vec{\tau}$-run.
We claim that there is a sequence $\auf\vec{c}(n) : n<B\zu$ of $N$-tuples of natural numbers
such that, for every $n<B$,
\begin{itemize}
\item[(a)]
$c^\lrun_i(n)\leq c_i(n)\leq c^\irun_i(n)$ for every $i<N$,
\item[(b)]
if $n>0$ then $\auf q_{n-1},\vec{c}(n-1)\zu\stepin\auf q_n,\vec{c}(n)\zu$.
\end{itemize}
It would follow that $\bigl\auf\auf q_n,\vec{c}(n)\zu : n<B\bigr\zu$ is a reliable $\vec{\tau}$-run
as required. 

We prove the claim by induction on $n$. To begin with, we let $\vec{c}(0):=\vec{0}$. Now suppose
that (a) and (b) hold for all $k< n$ for some $n$ with $0<n<B$. For each $i<N$, we let
\[
c_i(n):=\left\{
\begin{array}{ll}
c_i(n-1)+1, & \mbox{if $\alpha_n=\finci$},\\
c_i(n-1)-1, & \mbox{if $\alpha_n=\fdeci$},\\
c_i(n-1), & \mbox{if $\alpha_n=\ftest$ or $\alpha_n\in\{\fincj,\fdecj,\ftestj\}$ for $j\ne i$}.
\end{array}
\right.
\]
We need to check that (a) and (b) hold for $n$. There are several cases, depending on $\alpha_n$.
If $\alpha_n=\ftest$ then, by $\auf q_{n-1},\vec{c}^{\;\lrun}(n-1)\zu\lstepin\auf q_n,\vec{c}^{\;\lrun}(n)\zu$, the IH(a), and $\auf q_{n-1},\vec{c}^{\;\irun}(n-1)\zu\istepin\auf q_n,\vec{c}^{\;\irun}(n)\zu$, we have
\[
c_j^\lrun(n)\leq c_j^\lrun(n-1)\leq c_j(n-1)=c_j(n)\leq c_j^\irun(n-1)\leq c_j^\irun(n)\quad
\mbox{for all $j\ne i$}.
\]
Also, $c_i^\irun(n-1)=0$ by $\auf q_{n-1},\vec{c}^{\;\irun}(n-1)\zu\istepin\auf q_n,\vec{c}^{\;\irun}(n)\zu$.
So by the IH(a), we have $c_i(n-1)=0$, and so 
$c_i(n)=0$ and $\auf q_{n-1},\vec{c}(n-1)\zu\stepin\auf q_n,\vec{c}(n)\zu$.
As $\auf q_{n-1},\vec{c}^{\;\lrun}(n-1)\zu\lstepin\auf q_n,\vec{c}^{\;\lrun}(n)\zu$, we have
$c_i^\lrun(n)=0$. Thus
$c_i^\lrun(n)=c_i(n)=c_i^\irun(n-1)=0\leq c_i^\irun(n)$, as required.
The other cases are straightforward and left to the reader.
\end{proof}

In each of our lower bound proofs
we will use `faulty approximation', together with one of the following problems
on reliable counter machine runs:

\medskip
\noindent
\underline{{\sc CM non-termination:}} ($\Pi_1^0$-hard \cite{Minsky67})\\[3pt]
Given a counter machine $\mathcal{M}$, does $\mathcal{M}$ have an infinite reliable run?

\medskip
\noindent
\underline{{\sc CM reachability:}} ($\Sigma_1^0$-hard \cite{Minsky67})\\[3pt]
Given a counter machine $\mathcal{M}$, and a state $q_{\fin}$, does $\mathcal{M}$ have a reliable run reaching $q_{\fin}$?

 \medskip
\noindent
\underline{{\sc CM recurrence:}} ($\Sigma_1^1$-hard \cite{Alur&Henzinger})\\[3pt]
Given a counter machine $\mathcal{M}$ and a state $q_r$, does $\mathcal{M}$ have
a reliable run that visits $q_r$ infinitely often?


\section{Undecidable $\delta$-products with a $\K$-component}\label{kprod}

For each $0<k\leq\omega$, we call any frame $\auf k,R\zu$ a $k$-\emph{fan} if 
\begin{equation}\label{fan}
\{\auf 0,n\zu : 0<n<k\}\subseteq R.
\end{equation}

\begin{theorem}\label{t:kundec}
Let $L$ be any Kripke complete logic having an $\omega$-fan
among its frames. Then $\K\dprod L$ is undecidable.
\end{theorem}

\begin{corollary}\label{co:kundec}
$\K\dprod\Sfive$ is undecidable.
\end{corollary}

We prove Theorem~\ref{t:kundec} 
by  reducing the `CM non-termination' problem to $L_h\dprod L_v$-satisfiability.
Let $\M$ be a model based on the $\delta$-product of some frame $\F_h=\auf \Wh,\Rh\zu$ in $\Fr L_h$ and 
some frame $\F_v=\auf \Wv,\Rv\zu$ in $\Fr L_v$. First, we generate an $\omega\times\omega$-grid in $\M$.
Let $\gridfw$ be the conjunction of the formulas
\begin{align}
\label{hgen}
& \Uv\Dh\delta,\\
\label{vgen}
& \Bh\Dv(\Dh\delta\land\Bh\delta).
\end{align}

\begin{claim}\label{c:grid} {\bf (grid generation)}\\
If $\M,\auf r_h,r_v\zu\models\gridfw$ then there exist points
$\auf x_n\in \Wh\cap \Wv : n <\omega \zu$ such that, for all $n<\omega$,
\begin{itemize}
\item[{\rm (i)}] 
$r_h \Rh x_n$,
\item[{\rm (ii)}]  
$x_0=r_v$, and if $n>0$ then $x_0 \Rv x_n$,
\item[{\rm (iii)}]  
if $n>0$ then $x_{n-1}\Rh x_n$, 
\item[{\rm (iv)}]  
if $n>0$ then $x_n$ is the only  $\Rh$-successor of $x_{n-1}$.
\end{itemize}
(We do not claim that all the $x_n$ are distinct.)
\end{claim}

\begin{proof}
By induction on $n$. Let $x_0:= r_v$. Then (i) holds by \eqref{hgen}. Now suppose inductively
that we have $\auf x_k : k<n\zu$ satisfying (i)--(iv) for some $0<n<\omega$.
Then by \eqref{vgen}, there is $x_n\in \Wv$ such that $x_0\Rv x_n$ and $\M,\auf x_{n-1},x_n\zu\models\Dh\delta\land\Bh\delta$.
Therefore, $x_n\in \Wh$, $x_{n-1}\Rh x_n$, and $x_n$ is the only  $\Rh$-successor of $x_{n-1}$.
By \eqref{hgen}, $\M,\auf r_h,x_n\zu\models\Dh\delta$. So $r_h \Rh x_n$ follows, as required.
\end{proof}

Observe that because of Claim~\ref{c:grid}(iii) and (iv), $\Bh$ in fact expresses `horizontal next-time' in
our grid. For any formula $\psi$ and any $w\in W_v$, 
\begin{equation}\label{nexttime}
\M,\auf x_n,w\zu\models\Bh\psi
\qquad\mbox{iff}\qquad
\M,\auf x_{n+1},w\zu\models\psi,\quad\mbox{for all $n<\omega$}.
\end{equation}
Using this,
we will force a pair of infinite lossy and \ierror\ $\vec{\tau}$-runs,
for the same sequence $\vec{\tau}$ of instructions. 
Given any counter machine $M$,  for each $i<N$ of its counters, we take 
two fresh propositional variables $\cminus$ and $\cplus$. At each moment $n$ of time, the actual
content of counter $c_i$ during the lossy run will be represented by the set of points
\[
\OnSetd{i}(n):=\{w\in \Wv : x_0 \Rvr w\mbox{ and }\M,\auf x_n,w\zu\models\cminus\},
\]
and during the \ierror\ run by the set of points
\[
\OnSetu{i}(n):=\{w\in \Wv : x_0 \Rvr w\mbox{ and }\M,\auf x_n,w\zu\models\cplus\}.
\]
For each $i<N$, the following formulas force the possible changes in the counters during the
lossy and \ierror\ runs, respectively:
\begin{align*}
\fixd_i & : = \  \Uv(\Bh\cminus\to\cminus),\\
\incd_i & : = \ \Uv\bigl(\Bh\cminus\to (\cminus\lor\delta)\bigr),\\
\decd_i & : = \ \Uv(\Bh\cminus\to\cminus)\land\Ev(\cminus\land\Bh\neg\cminus),
\end{align*} 
and
\begin{align*}
\fixu_i & :=\  \Uv(\cplus\to\Bh\cplus),\\
\incu_i & :=\  \Uv(\cplus\to\Bh\cplus)\land\Ev(\neg\cplus\land\Bh\cplus),\\
\decu_i & :=\  \Uv\bigl(\cplus\to (\Bh\cplus\lor\delta)\bigr).
\end{align*}

\begin{claim}\label{c:counting} {\bf (lossy and \ierror\ counting)}\\
Suppose that $\M,\auf r_h,r_v\zu\models\gridfw$. Then for all $n<\omega$ and $i<N$:
\begin{itemize}
\item[{\rm (i)}] 
If $\M,\auf x_n,x_0\zu\models\fixd_i$ then $\OnSetd{i}(n+1)\subseteq \OnSetd{i}(n)$.
\item[{\rm (ii)}] 
If $\M,\auf x_n,x_0\zu\models\incd_i$ then $\OnSetd{i}(n+1)\subseteq \OnSetd{i}(n)\cup \{x_n\}$.
\item[{\rm (iii)}] 
If $\M,\auf x_n,x_0\zu\models\decd_i$ then $\OnSetd{i}(n+1)\!\subseteq \OnSetd{i}(n)- \{z\}$ for 
some $z\in \OnSetd{i}(n)$.
\item[{\rm (iv)}] 
If $\M,\auf x_n,x_0\zu\models\fixu_i$ then $\OnSetu{i}(n+1)\supseteq \OnSetu{i}(n)$.
\item[{\rm (v)}] 
If $\M,\auf x_n,x_0\zu\models\incu_i$ then there is $z$ such that $x_0 \Rvr z$, 
$z\notin \OnSetu{i}(n)$, and
$\OnSetu{i}(n+1)\supseteq \OnSetu{i}(n)\cup \{z\}$.
\item[{\rm (vi)}] 
If $\M,\auf x_n,x_0\zu\models\decu_i$ then $\OnSetu{i}(n+1)\supseteq \OnSetu{i}(n)- \{x_n\}$.
\end{itemize}
\end{claim}

\begin{proof}
We show items (ii) and(v). The proofs of the other items are similar and left to the reader.

(ii): Suppose $w\in \OnSetd{i}(n+1)$. Then 
$x_0 \Rvr w$ and $\M,\auf x_{n+1},w\zu\models\cminus$.
By \eqref{nexttime},
we have $\M,\auf x_{n},w\zu\models\Bh\cminus$.
Therefore, $\M,\auf x_{n},w\zu\models\cminus\lor\delta$ by $\incd_i$, and so either
$w\in \OnSetd{i}(n)$ or $w=x_n$.

(v): By $\incu_i$, there is $z$ with $x_0 \Rvr z$ and
$\M,\auf x_n,z\zu\models\neg\cplus\land\Bh\cplus$. Thus $z\notin \OnSetu{i}(n)$. Also, we have 
$\M,\auf x_{n+1},z\zu\models\cplus$ 
by \eqref{nexttime},
and so $z\in \OnSetu{i}(n+1)$.
Now suppose $w\in \OnSetu{i}(n)$. Then $x_0 \Rvr w$ and
$\M,\auf x_{n},w\zu\models\cplus$. By $\incu_i$, we have $\M,\auf x_{n},w\zu\models\Bh\cplus$.
Thus $\M,\auf x_{n+1},w\zu\models\cplus$ 
by \eqref{nexttime},
and so $w\in \OnSetu{i}(n+1)$.
\end{proof}

Using the above counting machinery, we can encode lossy and \ierror\ steps.
For each $\alpha\in\textit{Op}_C$, we define 
\[
\lexei:=\ \left\{
\begin{array}{ll}
\displaystyle\incd_i\land\bigwedge_{i\ne j<N}\fixd_j, & \mbox{ if $\alpha=\finci$},\\
\displaystyle\decd_i\land\bigwedge_{i\ne j<N}\fixd_j, & \mbox{ if $\alpha=\fdeci$},\\
\displaystyle\Uv\Bh\neg\cminus\land\bigwedge_{i\ne j<N}\fixd_j, & \mbox{ if $\alpha=\ftest$},\\
\end{array}
\right.
\]
and
\[
\iexei:=\ \left\{
\begin{array}{ll}
\displaystyle\incu_i\land\bigwedge_{i\ne j<N}\fixu_j, & \mbox{ if $\alpha=\finci$},\\
\displaystyle\decu_i\land\bigwedge_{i\ne j<N}\fixu_j, & \mbox{ if $\alpha=\fdeci$},\\
\displaystyle\Uv\neg\cplus\land\bigwedge_{i\ne j<N}\fixu_j, & \mbox{ if $\alpha=\ftest$}.\\
\end{array}
\right.
\]
Now we can force runs of $M$ that start at $\qini$ with all-zero counters.
For each state $q\in Q$, we introduce a fresh propositional variable $\mathsf{S}_q$, and define 
\begin{equation}\label{uniq}
	\St_q := \mathsf{S}_q \land\!\!\! \bigwedge_{q\ne q'\in Q} \!\!\!\neg \mathsf{S}_{q'}.
\end{equation}
Let $\mfw$ be the conjunction of
\begin{align}
\label{mini}
& \Bh\Bigl(\delta\to\bigl(\St_{\qini}\land\Uv(\neg\cminus\land\neg\cplus)\bigr)\Bigr),\\
\label{mstep}
& \Bh\bigwedge_{q\in Q-H}\Bigl(\St_q\to\bigvee_{\auf\alpha,q'\zu\in I_q}\bigl(\Bh\St_{q'}\land\lexei\land\iexei\bigr)\Bigr),\\
\label{minf}
& \Bh\bigvee_{q\in Q-H}\St_q.
\end{align}

\begin{lemma}\label{l:run} {\bf (lossy and \ierror\ run-emulation)}\\
Suppose that $\M,\auf r_h,r_v\zu\models\gridfw\land\mfw$. Let
$q_0:=\qini$, and for all $i<N$, $n<\omega$, let $c_i^\lrun(n):=|\OnSetd{i}(n)|$ and 
\[
c_i^\irun(n):=\left\{
\begin{array}{ll}
c_i^\irun(n-1)+1, & \mbox{if $\OnSetu{i}(n)$ is infinite,}\\[3pt]
|\OnSetu{i}(n)|,  & \mbox{otherwise}.
\end{array}
\right.
\]
Then there exists an infinite sequence
$\vec{\tau}=\bigl\auf\auf\alpha_n,q_n\zu: 0<n<\omega\bigr\zu$ of instructions 
such that
\begin{itemize}
\item
$\bigl\auf\auf q_n,\vec{c}^{\;\lrun}(n)\zu : n<\omega\bigr\zu$ is a lossy $\vec{\tau}$-run of $M$, and
\item
$\bigl\auf\auf q_n,\vec{c}^{\;\irun}(n)\zu : n<\omega\bigr\zu$ is an \ierror\ $\vec{\tau}$-run of $M$.
\end{itemize}
\end{lemma}

\begin{proof}
We define $\bigl\auf\auf\alpha_n,q_n\zu: 0<n<\omega\bigr\zu$ by induction on $n$ such that
for all $0<n<\omega$,
\begin{itemize}
\item
$q_n\in Q-H$ and
$\M,\auf x_n,x_0\zu\models\St_{q_n}$,
\item
$\auf q_{n-1},\vec{c}^{\;\lrun}(n-1)\zu\lstepin\auf q_n,\vec{c}^{\;\lrun}(n)\zu$ and
$\auf q_{n-1},\vec{c}^{\;\irun}(n-1)\zu\istepin\auf q_n,\vec{c}^{\;\irun}(n)\zu$.
\end{itemize}
As $\vec{c}^{\;\lrun}(0)=\vec{c}^{\;\irun}(0)=\vec{0}$ by \eqref{mini}, the lemma will follow.

To this end, 
take some $n$ with $0<n<\omega$. Then we have $q_{n-1}\in Q-H$ and $\M,\auf x_{n-1},x_0\zu\models\St_{q_{n-1}}$,
by \eqref{mini} and \eqref{minf} if $n=1$, and by the IH if \mbox{$n>1$.}
Therefore, by Claim~\ref{c:grid}(i) and \eqref{mstep}, there is $\auf \alpha_n,q_n\zu\in I_{q_{n-1}}$ such that
$\M,\auf x_{n-1},x_0\zu\models \Bh\St_{q_n}\land\lexein\land\iexein$.
So $\M,\auf x_{n},x_0\zu\models \St_{q_n}$ by Claim~\ref{c:grid}(iii), and so
$q_n\in Q-H$ by Claim~\ref{c:grid}(i) and \eqref{minf}.
Using Claim~\ref{c:counting}(i)--(iii), it is easy to check that
$\auf q_{n-1},\vec{c}^{\;\lrun}(n-1)\zu\lstepin\auf q_n,\vec{c}^{\;\lrun}(n)\zu$.
Finally, in order to show that $\auf q_{n-1},\vec{c}^{\;\irun}(n-1)\zu\istepin\auf q_n,\vec{c}^{\;\irun}(n)\zu$,
we need to use Claim~\ref{c:counting}(iv)--(vi) and the following observation. 
As for each $i<N$ either $\OnSetu{i}(n-1)$ is infinite or $c_i^\irun(n-1)=|\OnSetu{i}(n-1)|$,
if $c_i^\irun(n-1)\ne 0$ then $\OnSetu{i}(n-1)\ne\emptyset$, and so
$\alpha_n\ne\ftest$ follows by $\M,\auf x_{n-1},x_0\zu\models\iexein$.
\end{proof}

For each $k\leq\omega$, let $\HH_k$ be the frame obtained from $\auf k,+1\zu$ by adding a
`spy-point', that is,
let $\HH_k:=\auf k+1,S_k\zu$, where
\begin{equation}\label{spyframes}
S_k = \{\auf k,n\zu : n<k\} \cup \{\auf n-1,n\zu : 0<n<k\}.
\end{equation}

\begin{lemma}\label{l:sound} {\bf (soundness)}\\
If $M$ has an infinite reliable run, then $\gridfw\land\mfw$ is satisfiable in a model
over $\HH_\omega\dprod\F$ for some $\omega$-fan $\F$.
\end{lemma}

\begin{proof}
Suppose that 
$\bigl\auf\auf q_n,\vec{c}(n)\zu : n<\omega\bigr\zu$ is a reliable $\vec{\tau}$-run of $M$,
for some sequence $\vec{\tau}=\bigl\auf\auf\alpha_n,q_n\zu: 0<n<\omega\bigr\zu$ of instructions.
We define a model $\Mfw=\auf\HH_\omega\dprod\F,\efw\zu$ as follows.
For each $q\in Q$, we let
\[
\efw(\mathsf{S}_q):=\{\auf n,0\zu : n<\omega,\ q_n=q\}.
\]
Further, for all $i<N$, $n<\omega$, we will define inductively the sets $\efw_n(\cminus)$ and
$\efw_n(\cplus)$, and then put
\[
\efw(\cminus):=\{\auf n,m\zu : m\in\efw_n(\cminus)\}\ \ \mbox{and}\ \ 
\efw(\cplus):=\{\auf n,m\zu : m\in\efw_n(\cplus)\}.
\]
To begin with, we let $\efw_0(\cminus)=\efw_0(\cplus):=\emptyset$, and
\[
\efw_{n+1}(\cminus):=\left\{
\begin{array}{ll}
\efw_n(\cminus)\cup\{ n\}, & \mbox{ if }\alpha_{n+1}=\finci,\\
\efw_n(\cminus)-\{\min \efw_n(\cminus)\},  & \mbox{ if }\alpha_{n+1}=\fdeci,\\
\efw_n(\cminus), & \mbox{ otherwise}.
\end{array}
\right.
\]
It is straightforward to check that
\begin{equation}\label{lok}
|\efw_{n}(\cminus)|=c_i(n)\ \ \mbox{and}\ \
\Mfw,\auf n,0\zu\models\lexeinp,\quad\mbox{for all $i<N,\ n<\omega$}.
\end{equation}
We need to be a bit more careful when defining $\efw_n(\cplus)$. As the formulas 
$\iexein$ permit decrementing 
the \ierror\ counters only at diagonal points, we must be sure that only previously
incremented points get decremented. To this end, for every $i<N$, we let
\begin{equation}
\label{lseqfirst}
\Lambda_i  := \{ k<\omega : \alpha_{k+1}=\fdeci\},\qquad
\Xi_i  := \{ k<\omega : \alpha_{k+1}=\finci\},
\end{equation}
and let
\begin{align}
& \auf\lambda^i_m : m<L_i\zu\mbox{ be the enumeration of $\Lambda_i$ in ascending order, and}\\
\label{lseqlast}
& \auf\xi^i_m : m<K_i\zu\mbox{ be the enumeration of $\Xi_i$ in ascending order,}
\end{align}
for some $L_i,K_i\leq\omega$.
As in a run only non-zero counters can be decremented and our run is reliable, we always have 
$L_i\leq K_i$,
and $\lambda_m^i> \xi_m^i$ for all $m<L_i$.
Then we let
\[
\efw_{n+1}(\cplus):=\left\{
\begin{array}{ll}
\efw_n(\cplus)\cup\{ \lambda_m^i\}, & \mbox{ if }\alpha_{n+1}=\finci,\ n=\xi_m^i,\\
& \hspace*{3.1cm}m<L_i,\\
\efw_n(\cplus)\cup \bigl\{\min\bigl(\omega-\efw_n(\cplus)\bigr)\bigr\}, & \mbox{ if }\alpha_{n+1}=\finci,\ n=\xi_m^i,\ \\
& \hspace*{2.1cm}L_i\leq m<K_i,\\
\efw_n(\cplus)-\{n\},  & \mbox{ if }\alpha_{n+1}=\fdeci,\\
\efw_n(\cplus), & \mbox{ otherwise}.
\end{array}
\right.
\]
We claim that if $\alpha_{n+1}=\fdeci$ then $n\in \efw_n(\cplus)$,
and so $|\efw_{n+1}(\cplus)|=|\efw_n(\cplus)|-1$.
Indeed, 
if $\alpha_{n+1}=\fdeci$ then $n=\lambda_m^i$ for some $m<L_i$.
So $\efw_{\xi_m^i+1}(\cplus)= \efw_{\xi_m^i}(\cplus)\cup\{\lambda_m^i\}$, and
so $n\in \efw_{\xi_m^i+1}(\cplus)$. It follows that $n\in  \efw_{k}(\cplus)$ for 
every $k$ with $\xi_m^i+1\leq k <n+1$, as required.

Now it is not hard to see that $|\efw_{n}(\cplus)|=c_i(n)$ and $\Mfw,\auf n,0\zu\models\iexeinp$, 
for all $i<N$ and $n<\omega$.
Using this and  \eqref{lok}, it is easy to check that
 $\Mfw,\auf\omega,0\zu\models\gridfw\land\mfw$.
\end{proof}

Now Theorem~\ref{t:kundec} follows from Prop.~\ref{p:approx},  Lemmas~\ref{l:run}
and \ref{l:sound}. 


\bigskip
Note that it is easy to generalise the proof to obtain undecidability of $\T\dprod L$
(where $\T$ is the unimodal logic of all reflexive frames),
by using a version of the `tick-' or `chessboard'-trick 
(see e.g.\ \cite{Spaan93,Reynolds&Z01,gkwz05a} for more details):
Take a fresh propositional variable $\tvar$, and define a new `horizontal' modal operator by setting, 
for all formulas $\phi$,
\begin{equation}\label{tickbox}
\Bhr\phi := \bigl(\tvar\to \Bh(\neg\tvar\to\phi)\bigr)\land \bigl(\neg\tvar\to \Bh(\tvar\to\phi)\bigr).
\end{equation}
Then replace each occurrence of $\Bh$ in the formula $\gridfw\land\mfw$ with $\Bhr$, and add the conjunct
\begin{equation}\label{tickformula}
\Bh\bigl((\tvar\leftrightarrow\Bv\tvar)\land(\neg\tvar\leftrightarrow\Bv\neg\tvar)\bigr).
\end{equation}
It is not hard to check that the resulting formula is $\T\dprod L$-satisfiable iff $M$ has an infinite reliable run.

Next, recall $k$-fans from \eqref{fan}, and the frames $\HH_k$ from \eqref{spyframes}.

\begin{theorem}\label{t:kundecfin}
Let $\Ch$ and $\Cv$ be any classes of frames such that 
\begin{itemize}
\item 
either $\Ch$ or $\Cv$ contains only finite frames,
\item
either $\HH_\omega\in\Ch$, or $\HH_k\in\Ch$ for every $k<\omega$,
\item 
either $\Cv$ contains an $\omega$-fan, or
$\Cv$ contains a $k$-fan for every $k<\omega$.
\end{itemize}
Then $\Log(\Ch\dprod \Cv)$ is not recursively enumerable.
\end{theorem}

\begin{proof}
We sketch how to modify the proof of Theorem~\ref{t:kundec} to obtain a
 reduction of the `CM reachability' problem to $\Ch\dprod\Cv$-satisfiability.
To begin with, observe that if we add the conjunct
\begin{equation}\label{infgrid}
\Bh\Uv\bigl(p\lor\delta\to\Bh(p\land\neg\delta)\bigr)
\end{equation}
to the formula $\gridfw$ defined in \eqref{hgen}--\eqref{vgen}, then the grid-points $x_n$ generated
in Claim~\ref{c:grid} are all different.
 Now we introduce a fresh propositional variable $\stop$, and let $\gridfwfin$ be the conjunction
 of \eqref{hgen}, \eqref{infgrid} and the following `finitary' version of \eqref{vgen}:
\begin{equation}\label{vgenfin}
\Bh\Dv\bigl(\stop\lor(\Dh\delta\land\Bh\delta)\bigr).
\end{equation}
Given any counter machine $M$ and a state $q_{\fin}$,
let $\mfwfin$ be obtained from $\mfw$ by replacing \eqref{minf} with
\[
 \Bh\bigvee_{q\in(Q-H)\cup \{q_{\fin} \}}\St_q.
 \]
It is not hard to see that  
$\gridfwfin\land\mfwfin\land\Bh(\Dv\stop\to\St_{q_{\fin}})$ is $\Ch\dprod\Cv$-satisfiable iff
 there is a reliable run of $M$ reaching $q_{\fin}$.
\end{proof}

Note that it is also possible to give another proof of Theorem~\ref{t:kundec} by doing everything
`\emph{backwards}'. The conjunction of the following formulas generates a grid backwards
in $\K\dprod L$-frames, and is used in \cite{frocos09} to show that these logics lack
the finite model property w.r.t.\ \emph{any} (not necessarily product) frames:
\begin{align*}
& \Dv\Dh(\delta\land\Bh\bot),\\
& \Bv\bigl(\Dh\delta\to\Dh(\neg\delta\land\Dh\delta\land\Bh\delta)\bigr),\\
& \Bh\Dv\delta.
\end{align*}
Then the conjunction of the following formulas emulates counter machine runs,
again by going backwards along the generated grid:
\begin{align*}
& \Bh\Bigl(\Bh\bot\to\bigl(\St_{\qini}\land\Uv(\neg\cminus\land\neg\cplus)\bigr)\Bigr),\\
& \Bh\bigwedge_{q\in Q-H}\Bigl(\Dh\St_q\to\bigvee_{\auf\alpha,q'\zu\in I_q}\bigl(\St_{q'}\land\lexeibw\land\iexeibw\bigr)\Bigr),\\
& \Bh\bigvee_{q\in Q-H}\St_q,
\end{align*}
where 
\begin{align*}
\lexeibw & :=\ \left\{
\begin{array}{ll}
\displaystyle\incdbw_i\land\bigwedge_{i\ne j<N}\fixdbw_j, & \mbox{ if $\alpha=\finci$},\\
\displaystyle\decdbw_i\land\bigwedge_{i\ne j<N}\fixdbw_j, & \mbox{ if $\alpha=\fdeci$},\\
\displaystyle\Uv\neg\cminus\land\bigwedge_{i\ne j<N}\fixdbw_j, & \mbox{ if $\alpha=\ftest$},\\
\end{array}
\right.\\[3pt]
%
\iexeibw &:=\ \left\{
\begin{array}{ll}
\displaystyle\incubw_i\land\bigwedge_{i\ne j<N}\fixubw_j, & \mbox{ if $\alpha=\finci$},\\
\displaystyle\decubw_i\land\bigwedge_{i\ne j<N}\fixubw_j, & \mbox{ if $\alpha=\fdeci$},\\
\displaystyle\Uv\Bh\neg\cplus\land\bigwedge_{i\ne j<N}\fixubw_j, & \mbox{ if $\alpha=\ftest$},\\
\end{array}
\right.\\[3pt]
\fixdbw_i & : = \  \Uv(\cminus\to\Bh\cminus),\\
\incdbw_i & : = \ \Uv\bigl(\cminus\to (\Bh\cminus\lor\delta)\bigr),\\
\decdbw_i & : =\  \Uv(\cminus\to\Bh\cminus)\land\Ev(\neg\cminus\land\Bh\cminus),\\
\fixubw_i & := \ \Uv(\Bh\cplus\to\cplus),\\
%
\incubw_i & := \ \Uv(\Bh\cplus\to\cplus)\land\Ev(\cplus\land\Bh\neg\cplus),\\
\decubw_i & := \ \Uv\bigl(\Bh\cplus\to (\cplus\lor\delta)\bigr),
\end{align*}
for $i<N$.


\section{Undecidable $\delta$-products with a `linear' component}\label{linprod}

\begin{theorem}\label{t:linundec} 
Let $L_h$ be any Kripke complete logic such that $L_h$ contains $\Kfourt$ and 
$\auf\omega,<\zu$ is a frame for $L_h$.
Let $L_v$ be any Kripke complete logic having an $\omega$-fan among its frames. Then 
$L_h\dprod L_v$  is undecidable.
\end{theorem}

\begin{corollary}\label{co:linundec}
$\Kfourt\dprod\Sfive$ and $\Kfourt\dprod\K$ are both undecidable.
\end{corollary}

We prove Theorem~\ref{t:linundec} 
by reducing the `CM non-termination' problem to $L_h\dprod L_v$-satis\-fiability.
Let $\M$ be a model based on the $\delta$-product of a frame $\F_h=\auf \Wh,\Rh\zu$ for $L_h$
(so $\Rh$ is transitive and weakly connected%
\footnote{A relation $R$ is called \emph{weakly connected\/} if
$\forall x,y,z\,\bigl(xRy\land xR z\to (\mbox{$y=z$}\lor yRz\lor zRy)\bigr)$.}%
), and
some frame $\F_v=\auf \Wv,\Rv\zu$ for $L_v$. First, we again generate an $\omega\times\omega$-grid in $\M$.
Let 
\[
\lingridfw:= \ \delta\land\Uh\Dv(\Dh\delta\land\Bh\Bh\neg\delta).
\]

\begin{claim}\label{c:lingrid} {\bf (grid generation)}\\
If $\M,\auf r_h,r_v\zu\models\lingridfw$ then there exist points
$\auf x_n\in \Wh\cap \Wv : n <\omega \zu$ such that,
for all $n<\omega$,
\begin{itemize}
\item[{\rm (i)}]  
$x_0=r_v$, and if $n>0$ then $x_0 \Rv x_n$,
\item[{\rm (ii)}]
if $n>0$ then $\M,\auf x_{n-1},x_n\zu\models\Dh\delta\land\Bh\Bh\neg\delta$,
\item[{\rm (iii)}]  
if $n>0$ then, for every $z$,  $x_{n-1}\Rh z$ implies that $z=x_n$ or $x_n\Rh z$,
\item[{\rm (iv)}] 
$x_0=r_h$ and $x_m \Rh x_n$ for all $m<n$.
\end{itemize}
\end{claim}

\begin{proof}
By induction on $n$. Let $x_0:= r_h$. As $\M,\auf r_h,r_v\zu\models\delta$, we have $r_h=r_v$.
Now suppose inductively
that we have $\auf x_k : k<n\zu$ satisfying (i)--(iv) for some $0<n<\omega$.
Then there is $x_n\in \Wv$ such that $x_0\Rv x_n$ and $\M,\auf x_{n-1},x_n\zu\models\Dh\delta\land\Bh\Bh\neg\delta$.
Therefore, $x_n\in \Wh$, $x_{n-1}\Rh x_n$, and 
for every $z$,  $x_{n-1}\Rh z$ implies that $z=x_n$ or $x_n\Rh z$, by the
weak connectedness of $\Rh$.
So by the IH and the transitivity of $\Rh$,
we have $x_m \Rh x_n$ for all $m<n$.
\end{proof}

Next,
given any counter machine $M$, we will again force both an infinite lossy and an infinite 
\ierror\ $\vec{\tau}$-run, for the same sequence $\vec{\tau}$ of instructions. 
As $\Rh$ is transitive, we do not have a general `horizontal next-time' operator in our grid,
like we had in \eqref{nexttime}. However, because of Claim~\ref{c:lingrid}(iii) and (iv), we still
can have the following: For any formula $\psi$ and any $w\in W_v$, 
\begin{multline}\label{linnexttime}
\mbox{if $\psi$ is such that $\M,\auf x_{n+1},w\zu\models\psi\to\Bh\psi$, then}\\
\M,\auf x_n,w\zu\models\Bh\psi
\qquad\mbox{iff}\qquad
\M,\auf x_{n+1},w\zu\models\psi,\quad\mbox{for all $n<\omega$}.
\end{multline}
In order to utilise this,
 for each counter $i<N$ of $M$, we introduce two pairs 
of propositional variables: $\linminus$, $\loutminus$ for emulating lossy behaviour, and 
$\linplus$, $\loutplus$ for emulating \ierror\
behaviour. The following formula ensures that the condition in \eqref{linnexttime} hold for each of these
variables, at all the relevant points in $\M$:
\begin{multline*}
 \couf:=\ 
\bigwedge_{i<N}\Uh\Uv\bigl((\linminus\to\Bh\linminus)\land
(\loutminus\to\Bh\loutminus)\\
\land(\linplus\to\Bh\linplus)\land
(\loutplus\to\Bh\loutplus)\bigr).
\end{multline*}
At each moment $n$ of time, the actual
content of counter $c_i$ during the lossy run will be represented by the set of points
\[
\linOnSetd(n):=\{w\in \Wv : x_0 \Rvr w\mbox{ and }\M,\auf x_n,w\zu\models
\linminus\land\neg\loutminus\},
\]
and during the \ierror\ run by the set of points
\[
\linOnSetu(n):=\{w\in \Wv : x_0 \Rvr w\mbox{ and }\M,\auf x_n,w\zu\models
\linplus\land\neg\loutplus\}.
\]
For each $i<N$, the following formulas force the possible changes in the counters during the
lossy and \ierror\ runs, respectively:
\begin{align*}
\linfixd_i & : = \ \Uv(\Bh\linminus\to\linminus),\\
\linincd_i & : =\  \Uv\bigl(\Bh\linminus\to (\linminus\lor\delta)\bigr),\\
\lindecd_i & : = \ \Uv(\Bh\linminus\to\linminus)\land\Ev(\linminus\land\neg\loutminus\land\Bh\loutminus),
\end{align*} 
and
\begin{align*}
\linfixu_i & :=\  \Uv(\Bh\loutplus\to\loutplus),\\
\linincu_i & :=\  \Uv(\Bh\loutplus\to\loutplus)\land\Ev(\neg\linplus\land\neg\loutplus\land\Bh\linplus),\\
\lindecu_i & :=\  \Uv\bigl(\Bh\loutplus\to (\loutplus\lor\delta)\bigr).
\end{align*}
\begin{claim}\label{c:lincounting} {\bf (lossy and \ierror\ counting)}\\
Suppose that $\M,\auf r_h,r_v\zu\models\lingridfw\land\couf$. Then for all $n<\omega$, $i<N$:
\begin{itemize}
\item[{\rm (i)}] 
If $\M,\auf x_n,x_0\zu\models\linfixd_i$ then $\linOnSetd(n+1)\subseteq \linOnSetd(n)$.
\item[{\rm (ii)}] 
If $\M,\auf x_n,x_0\zu\models\linincd_i$ then $\linOnSetd(n+1)\subseteq \linOnSetd(n)\cup \{x_n\}$.
\item[{\rm (iii)}] 
If $\M,\auf x_n,x_0\zu\models\lindecd_i$ then $\linOnSetd(n+1)\!\subseteq \linOnSetd(n)- \{z\}$ for 
some $z\in \linOnSetd(n)$.
\item[{\rm (iv)}] 
If $\M,\auf x_n,x_0\zu\models\linfixu_i$ then $\linOnSetu(n+1)\supseteq \linOnSetu(n)$.
\item[{\rm (v)}] 
If $\M,\auf x_n,x_0\zu\models\linincu_i$ then there is $z$ such that $x_0 \Rvr z$, 
$z\notin \linOnSetu(n)$, and
$\linOnSetu(n+1)\supseteq \linOnSetu(n)\cup \{z\}$.
\item[{\rm (vi)}] 
If $\M,\auf x_n,x_0\zu\models\lindecu_i$ then $\linOnSetu(n+1)\supseteq \linOnSetu(n)- \{x_n\}$.
\end{itemize}
\end{claim}

\begin{proof}
We show items (iii) and (vi). The proofs of the other items are similar and left to the reader.

(iii): By $\lindecd_i$, there is $z$ such that $x_0 \Rvr z$ and
\[
\M,\auf x_n,z\zu\models\linminus\land\neg\loutminus\land\Bh\loutminus.
\]
So $z\in \linOnSetd(n)$. Also, by Claim~\ref{c:lingrid}(iv), 
\begin{equation}\label{zout}
\M,\auf x_{n+1},z\zu\models\loutminus.
\end{equation}
Now suppose $w\in \linOnSetd(n+1)$.
Then $x_0 \Rvr w$ and
$\M,\auf x_{n+1},w\zu\models\linminus\land\neg\loutminus$.
Then $\M,\auf x_{n},w\zu\models\neg\loutminus$ by $\couf$ and Claim~\ref{c:lingrid}(iv),
and $\M,\auf x_{n},w\zu\models\Bh\linminus$ by $\couf$ and \eqref{linnexttime}.
So we have $\M,\auf x_{n},w\zu\models\linminus$ by $\lindecd_i$, and so 
$w\in \linOnSetd(n)$. Finally, $w\ne z$ by \eqref{zout}.

(vi): Suppose that $w\in \linOnSetu(n)- \{x_n\}$. Then $x_0 \Rvr w$ and
$\M,\auf x_n,w\zu\models\linplus\land\neg\loutplus\land\neg\delta$.
Then $\M,\auf x_{n+1},w\zu\models\linplus$ by $\couf$ and Claim~\ref{c:lingrid}(iv),
and $\M,\auf x_{n},w\zu\models\neg\Bh\loutplus$ by $\lindecu_i$.
Therefore, $\M,\auf x_{n+1},w\zu\models\neg\loutplus$ by $\couf$ and \eqref{linnexttime},
 and so we have $w\in \linOnSetu(n+1)$.
\end{proof}

For each $\alpha\in\textit{Op}_C$, we define 
\[
\linlexei:=\ \left\{
\begin{array}{ll}
\displaystyle\linincd_i\land\bigwedge_{i\ne j<N}\linfixd_j, & \mbox{ if $\alpha=\finci$},\\
\displaystyle\lindecd_i\land\bigwedge_{i\ne j<N}\linfixd_j, & \mbox{ if $\alpha=\fdeci$},\\
\displaystyle\Uv(\Bh\linminus\to\Bh\loutminus)\land\bigwedge_{i\ne j<N}\linfixd_j, & \mbox{ if $\alpha=\ftest$},\\
\end{array}
\right.
\]
and
\[
\liniexei:=\ \left\{
\begin{array}{ll}
\displaystyle\linincu_i\land\bigwedge_{i\ne j<N}\linfixu_j, & \mbox{ if $\alpha=\finci$},\\
\displaystyle\lindecu_i\land\bigwedge_{i\ne j<N}\linfixu_j, & \mbox{ if $\alpha=\fdeci$},\\
\displaystyle\Uv(\linplus\to\loutplus)\land\bigwedge_{i\ne j<N}\linfixu_j, & \mbox{ if $\alpha=\ftest$}.\\
\end{array}
\right.
\]
For each state $q\in Q$, we introduce a fresh propositional variable $\mathsf{S}_q$, and define 
the formula $\St_q$ as in \eqref{uniq}.
Let $\linmfw$ be the conjunction of $\couf$ and the following formulas:
\begin{align}
\label{linmini}
& \St_{\qini}\land\Uv(\neg\linminus\land\neg\loutminus\land\neg\linplus\land\neg\loutplus),\\
\nonumber
& \Uh\bigwedge_{q\in Q-H}\Bigl[\Ev\St_q\to
\bigvee_{\auf\alpha,q'\zu\in I_q}\Bigl(\linlexei\land\liniexei\,\land\\
\label{linmstep}
& \hspace*{4.8cm}
\Uv\bigl(\Dh\delta\land\Bh\Bh\neg\delta\to\Bh(\delta\to\St_{q'})\bigr)\Bigr)\Bigr],\\
\label{linminf}
& \Uh\Uv\bigl(\delta\to\bigvee_{q\in Q-H}\St_q\bigr).
\end{align}

\begin{lemma}\label{l:linrun} {\bf (lossy and \ierror\ run-emulation)}\\
Suppose that $\M,\auf r_h,r_v\zu\models\lingridfw\land\linmfw$. Let
$q_0:=\qini$, and for all $i<N$, $n<\omega$, let $c_i^\lrun(n):=|\linOnSetd(n)|$ and 
\[
c_i^\irun(n):=\left\{
\begin{array}{ll}
c_i^\irun(n-1)+1, & \mbox{if $\linOnSetu(n)$ is infinite,}\\[3pt]
|\linOnSetu(n)|,  & \mbox{otherwise}.
\end{array}
\right.
\]
Then there exists an infinite sequence
$\vec{\tau}=\bigl\auf\auf\alpha_n,q_n\zu: 0<n<\omega\bigr\zu$ of instructions 
such that
\begin{itemize}
\item
$\bigl\auf\auf q_n,\vec{c}^{\;\lrun}(n)\zu : n<\omega\bigr\zu$ is a lossy $\vec{\tau}$-run of $M$, and
\item
$\bigl\auf\auf q_n,\vec{c}^{\;\irun}(n)\zu : n<\omega\bigr\zu$ is an \ierror\ $\vec{\tau}$-run of $M$.
\end{itemize}
\end{lemma}

\begin{proof}
We define $\bigl\auf\auf\alpha_n,q_n\zu: 0<n<\omega\bigr\zu$ by induction on $n$ such that
for all $0<n<\omega$
\begin{itemize}
\item
$q_n\in Q-H$ and
$\M,\auf x_n,x_n\zu\models\St_{q_n}$,
\item
$\auf q_{n-1},\vec{c}^{\;\lrun}(n-1)\zu\lstepin\auf q_n,\vec{c}^{\;\lrun}(n)\zu$ and
$\auf q_{n-1},\vec{c}^{\;\irun}(n-1)\zu\istepin\auf q_n,\vec{c}^{\;\irun}(n)\zu$.
\end{itemize}
As $\vec{c}^{\;\lrun}(0)=\vec{c}^{\;\irun}(0)=\vec{0}$ by \eqref{linmini}, the lemma will follow.

To this end, 
take some $n$ with $0<n<\omega$. Then we have $q_{n-1}\in Q-H$ and 
$\M,\auf x_{n-1},x_{n-1}\zu\models\St_{q_{n-1}}$,
by \eqref{linmini} and \eqref{linminf} if $n=1$, and by the IH if \mbox{$n>1$.}
So by Claim~\ref{c:lingrid}(i), we have $\M,\auf x_{n-1},x_{0}\zu\models\Ev\St_{q_{n-1}}$.
Thus by Claim~\ref{c:lingrid}(iv) and \eqref{linmstep}, there is $\auf \alpha_n,q_n\zu\in I_{q_{n-1}}$ such that
$\M,\auf x_{n-1},x_0\zu\models \linlexein\land\liniexein$ and
\begin{equation}\label{nextstate}
\M,\auf x_{n-1},x_0\zu\models\Uv\bigl(\Dh\delta\land\Bh\Bh\neg\delta\to\Bh(\delta\to\St_{q'})\bigr).
\end{equation}
Now it is easy to check that
$\auf q_{n-1},\vec{c}^{\;\lrun}(n-1)\zu\lstepin\auf q_n,\vec{c}^{\;\lrun}(n)\zu$ holds,
using Claim~\ref{c:lincounting}(i)--(iii).
In order to show that 
$\auf q_{n-1},\vec{c}^{\;\irun}(n-1)\zu\istepin\auf q_n,\vec{c}^{\;\irun}(n)\zu$,
we need to use Claim~\ref{c:lincounting}(iv)--(vi) and the following observation. 
As for each $i<N$ either $\linOnSetu(n-1)$ is infinite or $c_i^\irun(n-1)=|\linOnSetu(n-1)|$,
if $c_i^\irun(n-1)\ne 0$ then $\linOnSetu(n-1)\ne\emptyset$, and so
$\alpha_n\ne\ftest$ follows by $\M,\auf x_{n-1},x_0\zu\models\liniexein$.
Finally, we have 
$\M,\auf x_{n},x_n\zu\models \St_{q_n}$ by \eqref{nextstate} and Claim~\ref{c:lingrid}(ii),(iv),
and so
$q_n\in Q-H$ by Claim~\ref{c:lingrid}(i),(iv) and \eqref{linminf}.
\end{proof}

\begin{lemma}\label{l:linsound} {\bf (soundness)}\\
If $M$ has an infinite reliable run, then $\lingridfw\land\linmfw$ is satisfiable in a model
over $\auf\omega,<\zu\dprod\F$ for some countably infinite one-step rooted frame $\F$.
\end{lemma}

\begin{proof}
We may assume
that $\F=\auf \omega, S\zu$ and $\{\auf 0,n\zu : 0<n<\omega\}\subseteq S$.
Suppose that 
$\bigl\auf\auf q_n,\vec{c}(n)\zu : n<\omega\bigr\zu$ is a reliable run of $M$,
for some sequence $\vec{\tau}=\bigl\auf\auf\alpha_n,q_n\zu: 0<n<\omega\bigr\zu$ of instructions.
We define a model
\[ 
\linMfw=\bigl\auf\auf\omega,<\zu\dprod\F,\linefw\bigr\zu
\]
 as follows.
For each $q\in Q$, we let
\[
\linefw(\mathsf{S}_q):=\{\auf n,n\zu : n<\omega,\ q_n=q\}.
\]
Further, for all $i<N$, $n<\omega$, we will define inductively the sets 
$\linefw_n(\linminus)$, $\linefw_n(\loutminus)$, $\linefw_n(\linplus)$, and $\linefw_n(\loutplus)$,
and then put
\[
\linefw(\pvar):=\{\auf n,m\zu : m\in\linefw_n(\pvar)\},
\]
for $\pvar\in\{\linminus,\loutminus,\linplus,\loutplus\}$.
To begin with, we let $\linefw_0(\linminus)=\linefw_0(\loutminus)=\linefw_0(\linplus)=
\linefw_0(\loutplus):=\emptyset$, and
\begin{align*}
\linefw_{n+1}(\linminus) & :=\left\{
\begin{array}{ll}
\linefw_n(\linminus)\cup\{ n\}, & \mbox{ if }\alpha_{n+1}=\finci,\\
\linefw_n(\linminus), & \mbox{ otherwise},
\end{array}
\right.\\
\linefw_{n+1}(\loutminus) & :=\left\{
\begin{array}{ll}
\linefw_n(\loutminus)\cup\{\min\bigl(\linefw_n(\linminus)\!-\!\linefw_n(\loutminus)\bigr)\},  & \mbox{if }\alpha_{n+1}=\fdeci,\\
\linefw_n(\loutminus), & \mbox{ otherwise},
\end{array}
\right.\\
\linefw_{n+1}(\loutplus) & :=\left\{
\begin{array}{ll}
\linefw_n(\loutplus)\cup\{ n\}, & \mbox{ if }\alpha_{n+1}=\fdeci,\\
\linefw_n(\loutplus), & \mbox{ otherwise}.
\end{array}
\right.
\end{align*}
Next, recall the notation introduced in \eqref{lseqfirst}--\eqref{lseqlast}. We let
\[
\linefw_{n+1}(\linplus):=\left\{
\begin{array}{ll}
\linefw_n(\linplus)\cup\{ \lambda_m^i\}, & \mbox{if }\alpha_{n+1}=\finci,\ n=\xi_m^i,\\
& \hspace*{3cm}m<L_i,\\
\linefw_n(\linplus)\cup\bigl\{ \min\bigl(\omega-\linefw_n(\linplus)\bigr)\bigr\}, & \mbox{if }\alpha_{n+1}=\finci,\ n=\xi_m^i,\\
& \hspace*{2.1cm} L_i\leq m<K_i,\\
\linefw_n(\linplus), & \mbox{otherwise}.
\end{array}
\right.
\]
We claim that if $\alpha_{n+1}=\fdeci$ then $n\in \linefw_n(\cplus)=\linefw_{n+1}(\cplus)$,
and so 
\[
|\linefw_{n+1}(\linplus)-\linefw_{n+1}(\loutplus)|=
|\linefw_{n}(\linplus)-\linefw_{n}(\loutplus)|-1.
\]
Indeed, 
if $\alpha_{n+1}=\fdeci$ then $n=\lambda_m^i$ for some $m<L_i$.
So $\linefw_{\xi_m^i+1}(\linplus)= \linefw_{\xi_m^i}(\linplus)\cup\{\lambda_m^i\}$, and
so $n\in \linefw_{\xi_m^i+1}(\linplus)$. It follows that $n\in  \linefw_{k}(\linplus)$ for 
every $k$ with $\xi_m^i+1\leq k$. As $\lambda_m^i>\xi_m^i$, we have $n\in \linefw_n(\cplus)$ as required.

Now it is not hard to check that 
\[
|\linefw_n(\linminus)-\linefw_n(\loutminus)|=|\linefw_n(\linplus)-\linefw_n(\loutplus)|=c_i(n)
\]
 and
$\linMfw,\auf n,0\zu\models\linlexeinp\land\liniexeinp$,
for all $i<N$ and $n<\omega$, and so
$\linMfw,\auf 0,0\zu\models\lingridfw\land\linmfw$.
\end{proof}

Now Theorem~\ref{t:linundec} follows from Prop.~\ref{p:approx},  Lemmas~\ref{l:linrun}
and \ref{l:linsound}.


\bigskip
In some cases, we can have stronger lower bounds than in Theorem~\ref{t:linundec}.
We call a frame $\auf W,R\zu$ \emph{modally discrete} if it satisfies the following aspect of discreteness:
there are no points $x_0,x_1,\dots,x_n,\dots,x_\infty$ in $W$ such that
$x_0 R x_1 R x_2 R \dots Rx_nR\dots
Rx_\infty,\ x_n \ne x_{n+1}$ and $x_\infty \neg R x_n$,  for
all $n < \omega$.
We denote by $\disKfourt$ the logic  of all modally discrete linear orders.
Several well-known `linear' modal logics are extensions of
$\disKfourt$, for example, $\Log\auf\omega,<\zu$,
$\Log\auf\omega,\leq\zu$, $\GLt$ (the unimodal logic of all Noetherian%
\footnote{$\auf W,R\zu$ is
\emph{Noetherian} if it contains no infinite ascending chains $x_0
R x_1 R x_2 R \dots$ where $x_i \ne x_{i+1}$ for $i<\omega$.}
linear orders),
and $\Grzt$ (the unimodal logic of all Noetherian reflexive linear orders).
Unlike `real' discreteness, modal discreteness can be captured by modal formulas, and 
each of the logics above is finitely axiomatisable \cite{Segerberg70,Fine85}.

\begin{theorem}\label{t:linundecdisc} 
Let $L_h$ be any Kripke complete logic such that $L_h$ contains $\disKfourt$ and 
$\auf\omega,<\zu$ is a frame for $L_h$.
Let $L_v$ be any Kripke complete logic having an $\omega$-fan among its frames. Then 
both $L_h\dprod L_v$  and $L_h\sqprod L_v$ are $\Pi_1^1$-hard.
\end{theorem}

\begin{proof}
We sketch how to modify the proof of Theorem~\ref{t:linundec} to obtain a
 reduction of the `CM recurrence' problem to $L_h\dprod L_v$-satisfiability.
 Observe that by Claim~\ref{c:lingrid}(ii),(iv), the generated grid-points $x_n$ are 
 such that $x_n\ne x_{n+1}$ for all $n<\omega$.
 Therefore, if $\M$ is a model based on a \mbox{$\delta$-product} frame with a modally discrete 
 `horizontal'
 component and
 \[
 \M,\auf r_h,r_v\zu\models\lingridfw\land\linmfw\land\Bh\Dh\Dv(\delta\land\St_{q_r})
 \]
 for some state $q_{r}$,
then by Claim~\ref{c:lingrid}(iii),(iv), for every $n<\omega$ there is $k$ such that $n<k<\omega$ and $\M,\auf x_k,x_k\zu\models\St_{q_r}$.
\end{proof}

However, the formula $\lingridfw$ is clearly not satisfiable when $L_h$ has only
reflexive and/or dense frames (like $\Sfourt$, the unimodal logic of all reflexive linear orders, or 
the unimodal logic $\Log\auf\mathbb Q,<\zu$ over the rationals).
It is not hard to see that a `linear' version of the `tick-trick'  in
\eqref{tickbox}--\eqref{tickformula} can be used to generalise the proof of 
Theorem~\ref{t:linundec} for these cases.
Further, as by Claim~\ref{c:lingrid} the formula $\lingridfw$ forces an infinite ascending chain
of points, it is not satisfiable when $L_h$ has only Noetherian frames (like $\GLt$ or $\Grzt$). 
Similarly to the $\K$-case in Section~\ref{kprod},
it is also possible to generate an infinite grid and then emulate counter machine runs
by going \emph{backwards} in linear frames, and so 
to extend Theorem~\ref{t:linundec} to Noetherian cases.
The interested reader should consult  \cite{Hampson&Kurucz14}, where all
these issues are addressed in detail.


\section{Decidable $\delta$-products}\label{dec}

The following theorem shows that
the unbounded width of the second-component frames is
essential in obtaining the undecidability result of Theorem~\ref{t:kundec}:

\begin{theorem}\label{t:dec}
$L \dprod\Alt(n)$ is decidable in {\sc coNExpTime}, whenever $L$ is $\K$ or $\Alt(m)$,
for $0<n,m<\omega$.
\end{theorem}

\begin{proof}
We prove the theorem for $\K \dprod\Alt(n)$. The other cases are similar and left to the reader.
We show (by selective filtration) that if some formula $\phi$ does not belong to $\K \dprod\Alt(n)$,
then there exists a $\delta$-product frame for
$\K \dprod\Alt(n)$ whose size is exponential in $\phi$ where $\phi$ fails. 
It will also be clear that the presence or absence of the diagonal is irrelevant in our argument.

To begin with, we let
$\sub(\phi)$ denote the set of all subformulas of $\phi$. For any $\psi\in\sub(\phi)$, we denote
by $\hd(\psi)$ the maximal number of nested `horizontal' modal operators ($\Dh$ and $\Bh$) in $\psi$.
Similarly, $\vd(\psi)$ denotes the `vertical' nesting depth of $\psi$.
Now suppose that $\M,\auf r_h,r_v\zu\not\models\phi$ in some model $\M$ that is based on the
$\delta$-product of $\F_h=\auf\Wh,\Rh\zu$ and some frame $\F_v=\auf\Wv,\Rv\zu$ for $\Alt(n)$.
(Note that with $\delta$ in our language it is possible to force cycles in the component frames of a $\delta$-product, so
 we cannot assume that $\F_h$ and $\F_v$ are trees.)
For every $k\leq\vd(\phi)$, we define
\[
U_v^k:=\{y\in \Wv : \mbox{there is a $k$-long $\Rv$-path from $r_v$ to $y$}\}.
\] 
The $U_v^k$ are not necessarily disjoint sets for different $k$, but we always have
\begin{equation}\label{vbound}
|U_v^k|\leq 1+n+n^2+\dots+ n^k\leq 1+k\cdot n^k.
\end{equation}
Then
 we define $\F'_v:=\auf\Wv',\Rv'\zu$ by taking
\[
\Wv':=\bigcup_{k\leq\vd(\phi)}U_v^k,\hspace*{3cm}
\Rv':=\Rv\cap (\Wv'\times\Wv').
\]
Next, for every $m\leq\hd(\phi)$, we define inductively $U_h^m$ and $S_h^m$ as follows.
We let $U_h^0:=\{r_h\}$ and $S_h^0:=\emptyset$. Now suppose inductively that we have
defined $U_h^m$ and $S_h^m$ for some $m<\hd(\phi)$. For all $x\in U_h^m$, $y\in\Wv'$, and
$\Dh\psi\in\sub(\phi)$ with $\M,\auf x,y\zu\models\Dh\psi$, choose some $z_{x,y,\psi}$ from $\Wh$
such that $x\Rh z_{x,y,\psi}$ and $\M,\auf z_{x,y,\psi},y\zu\models\psi$. Then define
\begin{align*}
U_h^{m+1}& :=\{z_{x,y,\psi} : x\in U_h^m,\ y\in\Wv', \Dh\psi\in\sub(\phi),\ \M,\auf x,y\zu\models\Dh\psi\},\\
S_h^{m+1}&:=\{\auf x,z_{x,y,\psi}\zu : x\in U_h^m,\ y\in\Wv', \Dh\psi\in\sub(\phi),\ \M,\auf x,y\zu\models\Dh\psi\}.
\end{align*}
Again,
the $U_h^m$ are not necessarily disjoint sets for different $m$, but by \eqref{vbound} we always have that
\begin{equation}\label{hbound}
|U_h^m|\leq \bigl(\vd(\phi)\cdot n^{\vd(\phi)}\cdot |\sub(\phi)|\bigr)^m.
\end{equation}
Then we define $\F'_h:=\auf\Wh',\Rh'\zu$ by taking
\[
\Wh':=\bigcup_{m\leq\hd(\phi)}U_h^m,\hspace*{3cm}
\Rh':=\bigcup_{m\leq\hd(\phi)}S_h^m.
\]
Clearly, by \eqref{vbound} and \eqref{hbound} the size of $\F'_h\dprod\F'_v$ is exponential in the size of $\phi$. Let $\M'$ be the restriction of $\M$ to $\F'_h\dprod\F'_v$. Now a straightforward induction on $k$, $m$ and the structure of formulas shows that
 for all $k\leq\vd(\phi)$, $m\leq\hd(\phi)$, $\psi\in\sub(\phi)$,
 \[
 \M,\auf x,y\zu\models\psi\qquad\mbox{iff}\qquad
  \M',\auf x,y\zu\models\psi,
  \]
  whenever $x\in U_h^{\hd(\phi)-m}$, $y\in U_v^{\vd(\phi)-k}$, $\hd(\psi)\leq m$, and $\vd(\psi)\leq k$.
  It follows that $\M',\auf r_h,r_v\zu\not\models\phi$, as required. 
  \end{proof}

In certain cases the above proof gives polynomial upper bounds on the size of the falsifying
$\delta$-product model, so we have:

\begin{theorem}\label{t:altnp}
The validity problems of both $\Sfive\dprod\Alt(1)$ and $\Alt(1)\dprod\Alt(1)$  are {\sc coNP}-complete.
\end{theorem}

Note that all the above results hold with $\Alt(n)$ being replaced by its \emph{serial}%
\footnote{A frame $\auf W,R\zu$ is called \emph{serial\/}, if for every $x$ in $W$ there is $y$ with $xRy$.}
version
$\DAlt(n)$. One should simply make the `final' points in the filtrated
component frames reflexive.


\section{Open problems}\label{disc}

We have shown that in many cases adding a diagonal to product logics results in a dramatic increase
in their computational complexity (Sections~\ref{kprod} and \ref{linprod}), while in other cases upper bounds similar to diagonal-free product logics can be obtained (Section~\ref{dec}). Here are some related open problems:
 
\begin{enumerate}
\item
Theorems~\ref{t:linundec} and \ref{t:linundecdisc} 
do not apply when the first component logic has transitive but not necessarily weakly connected (linear) frames. In particular, while $\Kfour\mprod\Sfive$ is decidable
in {\sc coN2ExpTime} \cite{Gabbay&Shehtman98},
it is not known whether $\Kfour\dprod\Sfive$ remains decidable.
Note that it is not clear either whether we could somehow use Theorem~\ref{t:kundec} here, 
that is, whether
$\K\dprod\Sfive$ could be reduced to  $\Kfour\dprod \Sfive$. Note that 
the reduction of \cite{GollerJL12} from $\K\mprod L$ to $\Kfour\mprod L$ uses that 
$\K\mprod L$ is determined by product frames having intransitive trees as first components, and
this is no longer true for $\K\dprod L$. As is shown in Lemma~\ref{l:sound} and Claim~\ref{c:grid},
the formula $\gridfw$ defined in \eqref{hgen}--\eqref{vgen} is satisfiable in a $\delta$-product
frame for $\K\dprod L$, but forces a `horizontal' non-tree structure.

\item
By the above, $\K\dprod\K$ is properly contained in 
\[
\Log(\mbox{`Intransitive trees'}\dprod\mbox{`Intransitive trees'}),
\]
and Theorem~\ref{t:kundec} does not imply the undecidability of the latter. Is this logic decidable?
Note that it is not clear either whether the selective filtration proof of Theorem~\ref{t:dec} could be used here, as both component frames could be of arbitrary width. However, it might be possible to generalise one of the several proofs showing the decidability of $\K\mprod\K$
\cite{Gabbay&Shehtman98,gkwz03}.

\item
It can be proved using 2D type-structures called quasimodels that the diagonal-free product logic
$\K\mprod\Alt(1)$ is decidable in {\sc ExpTime} \cite[Thm.6.6]{gkwz03}. Is 
$\K\dprod\Alt(1)$ also decidable in {\sc ExpTime}?

\item
While
$\delta$-product logics are determined by $\delta$-product frames by definition, 
there exist other
(non-product, `abstract') $\delta$-frames for these logics. 
The \emph{finite frame problem} of a logic $L$ asks: ``Given a finite frame, is it a frame for $L$?''
If a logic $L$ is finitely axiomatisable, then its finite frame problem is of course decidable: one just has to check whether the finitely many axioms hold in the finite frame in question. However,
as is shown in  \cite{Kikot08}, many $\delta$-product logics ($\K\dprod\K$ and $\K\dprod\Kfour$ among them) are not finitely axiomatisable.
So the decidability of the finite frame problem is open for these logics. 
Note that if every finite frame for, say, $\K\dprod\K$ were the p-morphic image of a finite $\delta$-product frame, then we could enumerate finite frames for $\K\dprod\K$. As $\K\dprod\K$ is recursively enumerable by Theorem~\ref{t:re}, we can always enumerate those finite $\delta$-frames that are not frames for $\K\dprod\K$. So this would provide us with a decision algorithm for the finite frame problem of
$\K\dprod\K$.
However, consider the $\delta$-frame $\F=\auf W,\Rh,\Rv,D\zu$, where
\begin{align*}
& W=\{x,y,z\}, \hspace*{2cm} D=\{z\},\\
& \Rh =\{\auf x,x\zu,\auf y,y\zu,\auf z,z\zu,\auf y,z\zu,\auf z,x\zu,\auf y,x\zu\},\\
& \Rv =\{\auf x,x\zu,\auf y,y\zu,\auf z,z\zu,\auf x,z\zu,\auf z,y\zu,\auf x,y\zu\}.
\end{align*}
Then it is easy to see that $\F$ is a p-morphic image of 
$\auf\omega,\leq\zu\dprod\auf\omega\leq\zu$,
but $\F$ is not a p-morphic image of any finite $\delta$-product frame.
\end{enumerate}

\bibliographystyle{plain}


\end{document}